\documentclass[runningheads,oribibl]{llncs}
\usepackage{macros}
\usepackage{thm-restate}
\title{Non-Blind Strategies\\ in Timed Network Congestion Games\thanks{This work was partially funded by ANR project Ticktac (ANR-18-CE40-0015).}}
\author{Aline Goeminne \and Nicolas Markey \and Ocan Sankur
\\
\texttt{firstname.lastname@irisa.fr}
}
\authorrunning{Aline Goeminne \and Nicolas Markey \and Ocan Sankur}
\institute{Univ Rennes, Inria, CNRS, IRISA -- Rennes, France}

\begin{document}
\maketitle

\begin{abstract}
  \looseness=-1
Network congestion games are a convenient model for reasoning
about routing problems in a network: agents have to move from
a source to a target vertex while avoiding congestion, measured as a
cost depending on the number of players using the same link. Network
congestion games have been extensively studied over the last 40 years,
while their extension with timing constraints were considered more
recently.

Most of the results on network congestion games consider \emph{blind}
strategies: they are static, and do not adapt to the strategies
selected by the other players. We extend the recent results of
[Bertrand \emph{et~al.}, Dynamic network congestion games. FSTTCS'20]
to timed network congestion games, in which the availability of the
edges depend on (discrete) time. We~prove that computing Nash
equilibria satisfying some constraint on the total cost (and in
particular, computing the best and worst Nash equilibria),
and computing the social optimum,  can be
achieved in exponential space.
The~social optimum can be computed in polynomial space if all players have the same source and target.

\end{abstract}

\section{Introduction}
\paragraph{Network congestion games} allow one to model situations in which agents compete for resources
such as routes or bandwidth~\cite{Rosenthal1973}, e.g. in communication networks \cite{altman2009potential,QYZS-ieeeTN06}.
In these games, the objective of each agent is to go from a source to a target vertex, traversing a number of edges that represent resources.
The cost incurred by the player for the use of each resource is a function of the \emph{load}, that is, the number of agents using the same resource.
One of the fundamental notions studied in these games is that of \emph{Nash equilibria} which is used to model stable situations. 
A strategy profile is a Nash equilibrium if none of the players can reduce their cost by unilaterally changing their strategy.

It is well-known that these games can be \emph{inefficient} in the sense that there are Nash equilibria whose social cost (\IE, the sum of the agents' costs)
is bounded away from the optimum that can be achieved by arbitrary profiles (that may not be Nash equilibria).
Research has been focused on proving bounds on this inefficiency, formalized by the \emph{price of anarchy (PoA)}~\cite{KP-csr2009}.
A~tight bound of~$\frac{5}{2}$ on the price of anarchy was given in~\cite{awerbuch-stoc2005,Christodoulou2005b}.
The \emph{price of stability (PoS)} is dual to PoA: it is the ratio between the cost of the best Nash equilibrium and the social cost,
and was introduced in~\cite{Anshelevich2004}.
Bounds on PoA and PoS have been studied for restricted classes of graphs or types of cost functions~\cite{NisaRougTardVazi07}.

\paragraph{Timed network games.} Extensions of these games with real-time constraints have been
considered.  In~the setting of~\cite{hmrt-tcs11}, each edge is
traversed with a fixed duration, independent of its load, while the
cost is still a function of the load at the traversal time.  The~model
thus has \emph{time-dependent} costs since the load depends on the
times at which players traverse a given edge. The existence of Nash
equilibria is proven by reduction to~\cite{Rosenthal1973}.  The~work
in~\cite{AGK-mfcs17} propose another real-time extension, in which
transitions are instantaneous, and can only be taken during some
intervals of time. Time elapses in vertices, which incurs a cost that
is proportional to the load and the delay. In~those works, time is
continuous; \emph{boundary strategies} are strategies that take
transitions at dates that are boundaries of the constraining
intervals. It~is shown that \emph{boundary Nash equilibria} always
exist, but need not capture optimal Nash equilibria. The~prices of
anarchy and stability are shown to be bounded by $5/2$ and~$1+\sqrt 3/3$,
respectively, and computing \emph{some} Nash equilibrium
is \PLS-complete. This study was further extended to richer timing
constraints (involving clocks) in~\cite{AGK-mfcs18}.

\paragraph{Non-blind strategies.} In all the works mentioned above, the considered strategies are \emph{blind}: each player chooses a path to follow
at the beginning of the game, and cannot adapt their behaviors to their observations during the game.
\emph{Non-blind} strategies which allow players to choose their next moves depending on the whole history
were studied in~\cite{Bertrand20}. The advantage of non-blind players it that these allow one to obtain Nash equilibria
that have a lower social cost than with blind profiles. Thus, in general, the price of anarchy is lower with
non-blind strategies. In~\cite{Bertrand20}, the~existence of Nash equilibria was established for these strategies, 
and an algorithm was given to decide the existence of Nash equilibria satisfying given constraints on the costs of the players.

\paragraph{Our contributions.} We pursue the study of the real-time extensions of network congestion games
by considering non-blind strategies.
We~consider timed network congestion games similar to~\cite{AGK-mfcs17} (albeit with a discrete-time semantics), in~which the edges in the network
are guarded by a time interval which defines the time points at which the edges can be traversed.
Moreover, each vertex is endowed with a cost function depending on the load, and the players incur
a cost for each time unit spent on the vertex.
We consider both the \emph{symmetric} case in which all players source and target vertex pairs are identical,
and the \emph{asymmetric} case where these pairs vary.
We~formally define the semantics of our setting with non-blind strategies,
show how to compute the social optimum in \PSPACE in the symmetric case and in \EXPSPACE in the asymmetric case. %
Moreover, we~give an algorithm to decide the existence of Nash equilibria satisfying a given set of cost constraints 
in \EXPSPACE, for both the symmetric and asymmetric cases.
We can then compute the prices of anarchy and  of stability in \EXPSPACE.

\paragraph{Related Works.}
The existence of Nash equilibria in all atomic congestion games is proven using potential games.
The notion of potential was generalized by Monderer and Shapley~\cite{MS-geb96}
who showed how to iteratively use best-response computation to obtain a Nash equilibrium.
We refer the interested reader to~\cite{roughgarden-chap2007}
for an introduction and main results on the subject.

Timing constraints are also considered
in~\cite{ppt-mor2009,kp-icalp12} where travel times also depend on the
load. Other works focus on flow models with a timing
aspect~\cite{koch2011nash,bfa-geb2015}.

\section{Preliminaries}\label{sec-prelim}

\subsection{Timed network congestion game.}
\paragraph{Timed network.}
A~\emph{timed network} $\timedNet$ is a tuple $(\Vertex,\Edge,\weight,\guard)$
where
$\Vertex$ is a set of vertices,
$\Edge\subseteq V\times V$ is a set of edges,
$\weight\colon \Vertex \to (\bbN*\to \bbN*)$ associates with each vertex
  a non-decreasing weight function, and 
$\guard\colon \Edge \to \calP(\bbN)$ associates with each edge the
  dates at which that edge is available. We~require in the sequel that
for all $\edge\in \Edge$, $\guard(\edge)$~is a finite union of disjoint
intervals with bounds in~$\bbN\cup\{+\infty\}$.
A~guard~$\guard(\edge)$ is said to be satisfied at date~$\dat$
if $\dat\in\guard(\edge)$.

A~\emph{trajectory} in a timed network
$\timedNet=(\Vertex,\Edge,\weight,\guard)$ is a (finite or infinite) sequence
$(\vertex_j, \dat_j)_{0\leq j\leq l}$ such that for all~$0\leq j < l$,
$(\vertex_j,\vertex_{j+1})\in E$ and
$\dat_{j+1}\in\guard(\vertex_j,\vertex_{j+1})$.  We~write
$\Traj_{\timedNet}(\vertex,\dat,\vertex')$ for the set of trajectories with
$\vertex_0=v$, $\dat_0=\dat$, and $\vertex_l=v'$, possibly omitting to
mention~$\timedNet$ when no ambiguity arises.

\paragraph{Timed network game.}
A \emph{timed network congestion game} (a.k.a.~\emph{timed network
  game}, or~TNG for~short) is a model to represent and reason
about the congestion caused by the simultaneous use of some resources
by several users. 

\begin{definition}%
A \emph{timed network game} $\timedNetGame$ is a tuple
$(\nbp, \timedNet, (\src_i,\tgt_i)_{1\leq i\leq \nbp})$ where
\begin{itemize}
\item $\nbp$ is the number of players (in~binary). We~write $\Players$ for the set of players;
\item $\timedNet = (\Vertex,\Edge,\weight,\guard)$ is a timed network;
\item for each $1\leq i\leq \nbp$, the pair $(\src_i,\tgt_i)\in \Vertex^2$ is the
  objective of Player~$i$.
\end{itemize}
\end{definition}
\noindent
Symmetric TNGs are TNGs in which all players have the same
objective~$(\src,\tgt)$, i.e., $\src_i=\src$ and $\tgt_i=\tgt$ for all $i \in \Players$.

\begin{remark}\label{remark-encoding}
Several encodings can be used for the description
of~$(\src_i,\tgt_i)_{1\leq i\leq \nbp}$, which may impact the size of
the input: for symmetric TNGs, the players' objectives
would naturally be given as a single pair~$(\src,\tgt)$.
For the asymmetric case, the players' objectives
could be given explicitly as a list of $(\src,\tgt)$-pairs (the~size
is then linear in~$n$ and logarithmic in~$|\Vertex|$), or~as a
function ${\Vertex^2\to \bbN}$ (size~quadratic in~$|\Vertex|$ and
logarithmic in~$\nbp$).  Usually, the number of players is large (and
can be seen as a parameter), and the size of the input is at least
linear in~$|\Vertex|$, so that the latter representation is preferred.
\end{remark}

A \emph{configuration} of a TNG~$\timedNetGame$ is a mapping
${\config\colon \Players \to \Vertex}$ that maps each player to some vertex.
We~write $\Config$ for the set of all configurations of~$\timedNetGame$.
There are two particular configurations: the~\emph{initial
  configuration}~$\init$, in which all players are in their
source vertices,
\IE $\init\colon i \mapsto \src_i$;
and the final configuration $\final$ in which all
players are in their target vertices, \IE $\final\colon i \mapsto \tgt_i$.
A~timed configuration of~$\timedNetGame$ is a
pair~$(\config,\dat)$ where $\config$~is a configuration and
$\dat\in\bbN$~is the current time.
We~write
$\TConfig$ for the set of all timed configurations of~$\timedNetGame$,
and $\start$ for the starting timed configuration~$(\init,0)$.

\begin{remark}\label{rk-winningtraj}
We assume in the sequel that for each $0\leq i\leq \nbp$, there is a
trajectory from $(\src_i,0)$ to the target vertex $\tgt_i$;
this can be checked in polynomial time by exploring~$\timedNetGame$.
We also require that for each vertex~$\vertex\in\Vertex$ and each date~$\dat$, 
there exist a vertex~$\vertex'$
and a time~$\dat'>\dat$ such that $(\vertex,\vertex')\in\Edge$ and
$\dat'\in\guard(\vertex,\vertex')$. 
\end{remark}

\begin{example}
\label{ex:runEx}
 \begin{figure}[t]
\centering
\begin{tikzpicture}[scale=.9]
  \begin{scope}[minimum height=8mm]
  \node[draw,rounded corners = 5pt] (src) at (0,0){\nodeCong{1.5cm}{\src}{$x \mapsto 5x$}};
	\node[draw,rounded corners = 5pt] (s1) at (3,1.5){\nodeCong{1.5cm}{$s_1$}{$x \mapsto x$}};
	\node[draw,rounded corners = 5pt] (s2) at (6,1.5){\nodeCong{1.5cm}{$s_2$}{$x \mapsto 3x$}};
	\node[draw,rounded corners = 5pt] (tgt) at (9,0){\nodeCong{1.5cm}{\tgt}{$x \mapsto 1$}};
	\node[draw,rounded corners = 5pt] (s3) at (4.5,0){\nodeCong{2cm}{$s_3$}{$x \mapsto 10x+6$}};
	\node[draw,rounded corners = 5pt] (s4) at (3,-1.5){\nodeCong{1.5cm}{$s_4$}{$x \mapsto x$}};
	\node[draw,rounded corners = 5pt] (s5) at (6,-1.5){\nodeCong{1.5cm}{$s_5$}{$x \mapsto 3x$}};
  \end{scope}
	\draw[->,thick] (src) to[bend left] node[above]{$[2,3]$} (s1);
	\draw[->,thick] (s1) to node[above]{$[4]$} (s2);
	\draw[->,thick] (s2) to[bend left] node[above]{$[5]$} (tgt);
	\draw[->,thick] (src) to node[anchor=center,fill=white,pos=0.5]{$[1,2]$} (s3);
	\draw[->,thick] (s3) to node[anchor=center,fill=white,pos=0.5]{$[2,3]$} (tgt);
	\draw[->,thick] (src) to[bend right] node[below]{$[2,3]$} (s4);
	\draw[->,thick] (s4) to node[below]{$[4]$} (s5);
	\draw[->,thick] (s5) to[bend right] node[below]{$[5]$} (tgt);
	\draw[->,thick] (src) to node[anchor=center,fill=white,pos=0.5]{$[4]$} (s2);
	\draw[->,thick] (src) to node[anchor=center,fill=white,pos=0.5]{$[4]$} (s5);
	\draw[->,thick] (tgt) to[ loop right] (tgt);	
	\draw[->,thick] (src) to [loop left] node{} (src);
	
\end{tikzpicture}
\caption{Example of a timed network game}
\label{fig:runningEx}
\end{figure}
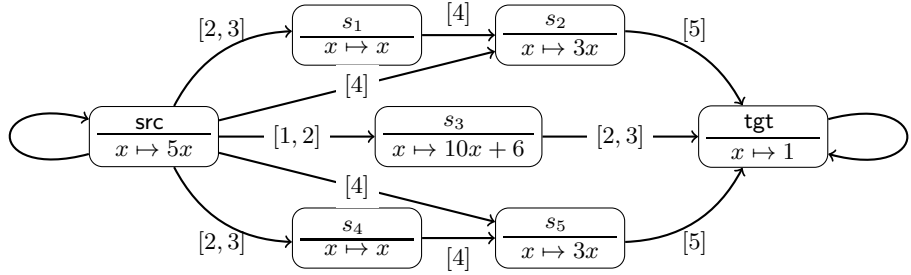

Figure~\ref{fig:runningEx} represents a timed network~$\timedNet$
(that we will use throughout this section to illustrate our
formalism). 
The~guards label the corresponding edges with two conventions: guards of the form $[d,d]$ for $d \in \N$ are denoted by $[d]$ and edges without label have $[0,+\infty)$ as a guard,  \emph{e.g.,} $\guard(\src,s_4) = [2,3]$ and $\guard(\src,\src) = [0+\infty)$.

The weight function $\wgt$ is defined inside the corresponding
vertex, \emph{e.g.,} $\wgt(\src)\colon x \mapsto 5x$ and
$\wgt(\tgt)\colon x \mapsto 1$.
As~explained in the sequel, this function indicates the cost of
spending one time unit in each vertex, depending on the number of
players. For instance, if~two players spend one time unit in~$\src$,
they both pay~$10$.
Notice that to comply with Remark~\ref{rk-winningtraj}, extra edges with guard $[6,+\infty)$ are added from each vertex $s_k$ ($1 \leq k \leq 5$) to an additional sink vertex (omitted in the figure for clarity).
\end{example}

\paragraph{Concurrent game associated with a timed network game.}
Consider a TNG $\timedNetGame=(\nbp, \timedNet,\penalty1000\relax (\src_i,\tgt_i)_{1\leq
  i\leq \nbp})$, with
  $\timedNet=(V,E,\wgt,\guard)$.  The~semantics runs
  intuitively as follows: a~timed configuration represents the positions of
  the players in the network at a given date. At~each round, each player selects
  the transition they want to take in~$\timedNet$,
  together with the date at which
  they want to take~it. All~players selecting the earliest date will
  follow the edges they selected, giving rise to a new timed
  configuration, from where the next round will take place.

We~express this semantics as an $n$-player infinite-state weighted concurrent
game~$\CGame=(\States, \Actions, (\Allowed_i)_{1\leq i\leq \nbp},
\Update,(\cost_i)_{1\leq i\leq \nbp})$~\cite{AHK02}.
The~set~$\States$ of states of~$\CGame$ is the
set~$\TConfig$ of timed configurations.  The~set~$\Actions$ of actions
is the set $\{ (\vertex,\dat) \mid \vertex\in\Vertex, \dat \in \bbN+ \}$.
Functions $\Allowed_i$ return the set of allowed actions for Player~$i$
from each configuration: for a timed configuration~$(\config,\dat)$,
we~have $(\vertex,\dat')\in\Allowed_i(\config,\dat)$
whenever $\dat'>\dat$ and
$(\config(i),\vertex)\in E$ and $\dat'\in\guard(\config(i),\vertex)$.
Notice that by Remark~\ref{rk-winningtraj},
$\Allowed_i(\config,\dat)$ is never empty.
An~\emph{action vector} in~$\CGame$ is a vector $\vec a = (a_i)_{1\leq
  i\leq \nbp}$ of actions, one for each player.
We~write $\Mov$ for the set of action vectors. An~action vector~$\vec
a$ is valid from state~$(\config,\dat)$ if for all~$1\leq i\leq \nbp$,
action~$a_i\in\Allowed_i(\config,\dat)$.  We~write
$\Valid(\config,\dat)$ for the set of valid action vectors from
state~$(\config,\dat)$.

From a state~$(\config,\dat)$, a~valid action vector $\vec
  a=(\vertex_i,\dat_i)_{1\leq i\leq\nbp}$ leads to a new
  state~$(\config',\dat')$ where $\dat'=\min\{\dat_i\mid 1\leq
  i\leq \nbp\}$ and for all $1\leq i\leq
\nbp$, $\config'(i)=\vertex_i$ if $\dat_i=\dat'$, and
$\config'(i)=\config(i)$ otherwise.
This~is encoded in the update
function~$\Update$ by letting $\Update((\config,\dat),\vec
a)=(\config',\dat')$.
This~models the fact that the
players proposing the shortest delay apply their action once this
shortest delay has elapsed.
We~write $\Select(\vec a) = \{1\leq i\leq\nbp\mid \forall 1\leq j\leq \nbp.
\ \dat_i\leq \dat_j\}$.

We let $\Trans = \{ (s,\vec a,s') \in \States \times \Mov \times \States \mid  \vec a \in \Valid(s),\,
    {\Update(s,\vec a) = s'} \}$ be the set of \emph{transitions} of~$\CGame$.
    We~write~$G$ for the graph structure~$(\States,\Trans)$.

\looseness=-1
The delays elapsed in the vertices of~$\timedNetGame$
come with a cost for each player, computed as follows:
given a configuration~$\config$, we~write $\load_{\config}\colon
\Vertex\to\bbN$ for the load of each vertex of the timed network,
i.e., the number of players sitting in that vertex:
$\load_{\config}(\vertex)= \#\{1\leq i\leq \nbp\mid
\config(i)=\vertex\}$. For each~$1\leq i\leq \nbp$, the~cost for
Player~$i$ of taking a transition~$t=((\config,\dat),\vec a,(\config',\dat'))$,
moving from timed configuration~$(\config,\dat)$ to
timed configuration~$(\config',\dat')$, then is
$\cost_i(t) =
{(\dat'-\dat)}\cdot \weight(\config(i))(\load_{\config}(\config(i)))$:
it~is proportional to the time elapsed and to the load of the vertex
where Player~$i$ is waiting. Notice that it does not depend on~$\vec
a$ (nor~on~$\config'$).

\begin{remark}
With our definition, players remain in the game even after they
have reached their targets. This may be undesirable in some
situations. One~way of avoiding this is to add another limitation on
the availability of transitions to players: transitions would be
allowed only to some players, and in particular, from their target
states, players would only be allowed to go to a sink state.
All~our results can be adapted to this setting.
\end{remark}

\begin{example}\label{ex-TNCG}
We illustrate those different notions on a two-player symmetric timed network
game $\timedNetGame=(2,\timedNet,(\src,\tgt))$,
where $\timedNet$ is the timed network of Example~\ref{ex:runEx}.

Let $(\config,\dat)=((\src,s_1),2)$ be the timed configuration in~$\TConfig$, with $\config(1)=\src$, $\config(2)= s_1$ and $\dat = 2$. From this timed configuration, the set of allowed actions for Player~$1$ is $\Allowed_1(\config,\dat) = \{ (s_1,3),(s_4,3),(s_2,4),(s_5,4),(\src,d) \mid d \geq 3\}$. If we consider the valid action vector $\vec a = ((s_1,3),(s_2,4))$, then $\Update((\config,\dat), \vec a ) = (\config',\dat')$ with $ \config' = (s_1,s_1)$ and $\dat'=3$. Indeed even if both players choose their actions simultaneously, we~have $\Select(\vec a) = \{ 1 \}$, and Player~$1$ is the only player who moves at time~$3$.
Moreover, writing $t = ((\config,\dat),\vec a, (\config',\dat'))$, we~have $\cost_1(t) = 5$, because $\load_{(\config,\dat)}(\src) = 1$. Finally, let us consider another example of cost computation: let $(\config,\dat)=((s_3,s_3),1)$ and $(\config',\dat')=((\tgt,\tgt),3)$, and $\vec a=((\tgt,3),(\tgt,3))$; we have that $\cost_1((\config,d),\vec a,(\config',d'))= 2 \cdot (10 \cdot 2 + 6) = 52$.
\end{example}

\paragraph{Plays and histories.}
A~play~$\rho=(s_k,\vec a_k, s'_k)_{k\in\bbN} \in \Trans^\omega$
in $\CGame$ (also called play in~$\timedNetGame$) is
an infinite sequence of transitions such that
for all $k \in \N$, $s'_k  = s_{k+1}$.
We~denote by~$\Plays$ the set of plays of~$\CGame$.
We~denote by
$\Plays(s)$ the set of plays that start in state $s \in
\States$, i.e. $\Plays(s) = \{ (s_k,\vec a_k,s'_k)_{k\geq 0} \in
\Plays \mid s_0 = s \}$.

A~history is a finite prefix of a play.  We~denote the set of
histories by $\Hist$, and the subset of histories starting in a given
state~$s$ by $\Hist(s)$.

\looseness=-1
Given a play $\rho = (s_k,\vec a_k,s'_k)_{k\geq 0} \in \Plays$ and an
integer~$j\in\bbN$, we~write~$\rho_j$
for the timed configuration~$s_j$, 
$\rho_{\geq j}$ for the suffix
  $(s_{j+k}, \vec a_{j+k}, s'_{j+k})_{k\geq 0}\in\Plays(s_j)$,
and $\rho_{<j}$ for the history $(s_{k}, \vec a_{k}, s'_k)_{k<j}\in \Hist(s_0)$.
For a history~$h=(s_k,\vec a_k,s'_k)_{0\leq k<j}$ in~$\Hist(s_0)$,
we~write $\last(h)=s'_{j-1}$ when~$j>0$ (we~may also write it as $s_j$ when no ambiguity arises), and $\last(h)=s_0$ otherwise.

\paragraph{Cost functions.}
For each player $i \in \Players$, we~define a cost function
$\cost_i\colon \Plays \to \N \cup \{+\infty \}$ such that for all
$\rho = (s_k,\vec a_k,s'_k)_{k\in\bbN} \in \Plays$,
\[
\cost_i(\rho) =
\begin{cases}
 +\infty & \text{ if $s_{k}(i)\not=\tgt_i$ for all~$k\in\bbN$} \\
\sum_{k=0}^{\ell-1} \cost_i(s_k,\vec a_k,s'_k)
 & \text{ if $\ell$ is the least index such that } s_{\ell}(i)=\tgt_i
\end{cases}
\]
This function is extended to histories in the natural way.
For all $\rho \in \Plays\cup\Hist$, the~vector
$\cost(\rho) = {(\cost_i(\rho))}_{1\leq i\leq\nbp}$
is the \emph{cost profile} of~$\rho$.

\paragraph{Strategies.} Given a concurrent game $\CGame$, a~state~$s$
in~$\CGame$, and $1\leq i\leq \nbp$,
a~\emph{strategy} for Player~$i$ from~$s$ is a function
$\sigma_i \colon \Hist(s) \to \Actions$ such that $\sigma_i(h)
\in \Allowed_i(\last(h))$ for all $h \in \Hist(s)$.
We~denote by~$\Sigma_i(s)$ the set of strategies of Player~$i$
from~$s$ (we~may omit to mention~$s$ when it is the initial timed
configuration~$\start$).
Given a state~$s$, a~subset~$I\subseteq \Players$ of players,
and strategies~$(\sigma_i)_{i\in I}$ from~$s$ for those players, a~play
$\rho=(s_k,\vec a_k,s'_k)_{k\in\N}$ from~$s$ is consistent with
$(\sigma_i)_{i\in I}$ if for all~$k\in\bbN$ and all $i\in I$, it~holds
$\vec a_{k,i}=\sigma_i(\rho_{<k})$.

A~\emph{strategy profile} $\sigma = {(\sigma_i)}_{1\leq i\leq \nbp}$ from~$s$ is
a tuple of strategies from~$s$, one for each player.
We~write $\Sigma(s)$ for the set of strategy profiles from~$s$.
Given a strategy
profile~$\sigma$ from~$s$, there is a unique play~$\rho$ from~$s$
that is \emph{consistent} with that strategy profile.
This~play is denoted by~$\InitOutcome{\sigma}{s}$ and is
called the \emph{outcome} of the strategy profile~$\sigma$ from~$s$.

Given a strategy profile~$\sigma$ from~$s$, a~player~$i\in\Players$ and
a strategy~$\sigma'_i$ from~$s$ for Player~$i$,
we~write $(\sigma_{-i},\sigma'_i)$ for the strategy profile
$(\tau_j)_{1\leq j\leq \nbp}$ for which $\tau_j=\sigma_j$ when $j\not=i$ and
$\tau_i=\sigma'_i$.
Given a player~$1\leq i\leq \nbp$
and a strategy~$\sigma_i$ from~$s$ for Player~$i$, for all $h \in
\Hist(s)$, we~denote by~$\sigma_{i \restriction h}$ the strategy
of Player~$i$ from~$\last(h)$
such that $\sigma_{i \restriction h} \colon h'\in
\Hist(\last(h)) \mapsto \sigma_i(hh')$.
This is extended to strategy profiles in the natural way.

\paragraph{Blind strategies.}
An~important class of strategies is the class of \emph{blind strategies}.
Intuitively, a~blind strategy follows a single trajectory,
without looking at how the other players are moving. In~order to
define blind strategies, we~first introduce a notion of \emph{projection}
of plays on the actions of each player.

For any play $\rho = {(s_k,\vec a_k,s'_k)}_{k\in\bbN}$ from~$\start$
and any $1\leq i\leq\nbp$, writing $\vec a_k=(\vertex_{k,i},\dat_{k,i})_{1\leq
i\leq \nbp}$ for all~$k\in\bbN$, we~define
$\mu_i\colon \bbN \mapsto \bbN\cup\{+\infty\}$
inductively as $\mu_i(0)=-1$ and, for
all~$j$ for which $\mu_i(j)<+\infty$,
\[
\mu_i(j+1) = \inf\{ k>\mu_i(j) \mid \dat_{k,i} = \min\{\dat_{k,l} \mid
  1\leq l\leq \nbp\}\}.
\]
In~other terms, $\mu_i$ returns the indices where Player~$i$ proposed
a minimal delay, and thus could move to the vertex they proposed.
The~trajectory of Player~$i$ along~$\rho$, denoted with
$\traj_i(\rho)$, is then defined as the trajectory
$(\vertex_{\mu_i(j),i}, \dat_{\mu_i(j),i})_{j\geq 0}$, with the convention that
$(\vertex_{-1,i}, \dat_{-1,i}) = (\src_i,0)$.
Notice that this trajectory could be finite.
Functions~$\traj_i$ are extended to histories in the natural way.

\begin{definition}
A strategy~$\sigma_i$ for Player~$i$ is \emph{blind} if, for any two
histories~$h$ and~$h'$ such that $\traj_i(h)=\traj_i(h')$, it~holds
$\sigma_i(h)=\sigma_i(h')$.
\end{definition}

\noindent\looseness=-1
As the next lemma suggests, playing a blind strategy amounts to
following a fixed trajectory in the timed network, independently of
the actions of the other players.

\begin{restatable}{lemma}{lemmablindpath}\label{lemma-blindpath}
Let $i \in \Players$ be a player and $\sigma_i$ be a
        strategy of Player~$i$ from some state~$s_0$.
        If~$\sigma_i$ is blind,
        then for all strategies $\sigma_{-i}$ and $\sigma'_{-i}$ from~$s_0$,
        we~have
	\[
        \traj_i(\Outcome{\sigma_{-i},\sigma_i}_{s_0}) =
          \traj_i(\Outcome{\sigma'_{-i},\sigma_i}_{s_0}).
        \]
\end{restatable}

In view of this lemma, any blind strategy $\sigma_i$ for Player~$i$ from some state~$s_0$
can be associated with a \emph{trajectory}, denoted with $\langle
\sigma_i \rangle_{s_0}$, and defined as
\(
\langle \sigma_i \rangle_{s_0} =
\traj_i(\Outcome{\sigma_{-i},\sigma_i}_{s_0})
\)
(for any $\sigma_{-i} \in \Sigma_{-i}(s_0)$.)
Conversely, 
for any trajectory~$\pi$ from $(v_0,d_0)$, any player~$i\in\Players$,
and any timed configuration $s_0=(\config,\dat_0)$ with $\config(i)=v_0$,
there exists a blind strategy~$\sigma_i$ for Player~$i$ from~$s_0$ whose
associated trajectory is~$\pi$.

A play $\rho=(s_k,\vec a_k,s'_k)_{k\in\bbN}$ from the initial timed
configuration~$\start$ is said \emph{winning} for Player~$i$ if there exists $k\in\bbN$ such that $s_k = (\config_k,\dat_k)$ with $\config_k(i) = \tgt_i$.
A~strategy~$\sigma_i$ for
Player~$i$ from~$\start$ is winning if any play
consistent with that strategy is winning for Player~$i$.  For
each~$1\leq i\leq \nbp$, as we assumed
$\Traj(\src_i,0,\tgt_i)\not=\varnothing$, we~get that there exists a
winning (blind) strategy~$\sigma_i$ for Player~$i$.

\begin{example}
We consider the two-player symmetric TNG described in Example~\ref{ex-TNCG}. The infinite sequence of transitions 
$ \rho\colon ((\src,\src),0) \xrightarrow{\left[\begin{smallmatrix}(s_1,2)\\(s_2,4)\end{smallmatrix}\right]} ((s_1,\src),2) \xrightarrow{\left[\begin{smallmatrix}(s_2,4)\\(s_2,4)\end{smallmatrix}\right]} ((s_2,s_2),4) \biggl( \xrightarrow{\left[\begin{smallmatrix}(\tgt,5+k)\\(\tgt,5+k)\end{smallmatrix}\right]} ((\tgt,\tgt),5+k) \biggr)_{k \geq 0}$ is a play in $\timedNetGame$ such that $\cost_1(\rho) = 2\cdot (5\cdot 2) + 2\cdot (1\cdot 1) + 1\cdot (3\cdot 2) = 28$ and $\cost_2(\rho)= 2 \cdot (5\cdot 2) + 2\cdot(5\cdot 1)+ 1\cdot(3\cdot 2) = 36$. 

Let us now consider two trajectories $\pi_1\colon (\src,0)(s_1,2)(s_2,4)(\tgt,5+k)_{k\geq 0}$ and $\pi_2 \colon (\src,0)(s_2,4)(\tgt,5+k)_{k\geq 0}$ together with $\sigma_1$, a blind strategy of Player~$1$, and $\sigma_2$, a blind strategy of Player~$2$, such that $\Outcome{\sigma_1}_{\start} = \pi_1$ and $\Outcome{\sigma_2}_{\start} = \pi_2$. The~outcome of the strategy profile $(\sigma_1,\sigma_2)$ from $\start$ is $\rho$, \emph{i.e.,} $\Outcome{\sigma_1,\sigma_2}_{\start} = \rho$.
\end{example}

\paragraph{Social optima and Nash equilibria.}
Let $\CGame$ be a concurrent game.
The \emph{social welfare} of a play~$\rho$ in~$\CGame$ is the sum of
the costs of the players along~$\rho$: $\SW(\rho) = \sum_{1\leq
i\leq \nbp} \cost_i(\rho)$.  The social welfare of a strategy profile
is the social welfare of its outcome.

A strategy profile~$\sigma$ from the starting configuration~$\start$
is a social optimum~(SO) whenever
\( \SW(\Outcome{\sigma}_{\start}) = 
\inf_{\sigma'\in\Sigma(\start)} \SW(\Outcome{\sigma'}_{\start}) 
\).
When no ambiguity arises, social optimum may also refer to the social
welfare of socially-optimal strategy profiles. Notice that since we
consider discrete time, this value is an integer, and a strategy
profile realizing the social optimum always exists.

A~strategy profile $\sigma$ from~$\start$
is a \emph{Nash equilibrium}~(NE) if for
all $1\leq i\leq\nbp$, for all strategies $\sigma'_i \in \Sigma_i(\start)$ of
Player~$i$,
$\cost_i(\InitOutcome{\sigma}{\start}) \leq
\cost_i(\InitOutcome{ \sigma_{-i}, \sigma'_{i}}{\start})$:
for~all~$1\leq i\leq\nbp$,
the~strategy~$\sigma_i$ of Player~$i$ is (one~of) the best
strategies against the strategies~$\sigma_{-i}$ of the other players.

\begin{example}
We show an example of a Nash equilibrium from $\start$ in the two-player symmetric TNG described in Example~\ref{ex-TNCG}. 
The~strategy of Player~$1$, denoted by $\sigma_1$, is given by a blind strategy associated with the trajectory $\pi_1= (\src,0)(s_3,1)(\tgt,2+k)_{k\geq 0}$.
The~strategy of Player~$2$, denoted by $\sigma_2$, is a  more involved. The~first action of Player~$2$ from the timed configuration $\start$ is $(\src,1)$, then at time~$1$ they observe if Player~$1$ has complied with their strategy, \emph{i.e.,} if Player~$1$ is in~$s_3$. If~so, Player~$2$ moves to~$s_4$ at time~$2$, to~$s_5$ at time~$4$, and ends in~$\tgt$ at time~$5$; otherwise, they wait in~$\src$ in order to observe the exact deviation of Player~$1$, and punish them adequately, by moving from $\src$ to~$s_5$ at time~$4$ if Player~$1$ is in~$s_4$ or to~$s_2$ otherwise, and ending in~$\tgt$ at time~$5$.
The~outcome $\Outcome{\sigma_1,\sigma_2}_{\start}$ of this strategy profile from $\start$ is:  %

\begin{multline*}
\rho \colon \start \xrightarrow{\left[\begin{smallmatrix}(s_3,1)\\
  (\src,1)\end{smallmatrix}\right]}((s_3,\src),1) \xrightarrow{\left[\begin{smallmatrix}(\tgt,2)\\(s_4,2)\end{smallmatrix}\right]}((\tgt,s_4),2) \xrightarrow{\left[\begin{smallmatrix}(\tgt,3)\\(s_5,4)\end{smallmatrix}\right]}((\tgt,s_4),3) \\
  \xrightarrow{\left[\begin{smallmatrix}(\tgt,4)\\(s_5,4)\end{smallmatrix}\right]}((\tgt,s_5),4) \biggl (\xrightarrow{\left[\begin{smallmatrix}(\tgt,5+k)\\(\tgt,5+k)\end{smallmatrix}\right]}((\tgt,\tgt),5+k) \biggr)_{k \geq 0}.
\end{multline*}

We have that $\Cost_1(\rho) = 26$ and $\cost_2(\rho) = 20$. In~particular, the social welfare of $\rho$ is $\SW(\rho) = 46$.

We can prove that $(\sigma_1,\sigma_2)$ is a Nash equilibrium. 
We only explain here the interest of the threat of punishment of Player~$2$ by considering the deviating strategy $\sigma'_1$ of Player~$1$ such that he moves from $\src$ to~$s_1$ at time~$2$, to~$s_2$ at time~$4$ and ends in~$\tgt$ at time~$5$.
The outcome $\Outcome{\sigma'_1,\sigma_2}_{\start}$ of this new strategy profile from $\start$ is 
\begin{multline*}\rho'\colon \start \xrightarrow{\left[\begin{smallmatrix}(s_1,2)\\(\src,1)\end{smallmatrix}\right]}((\src,\src),1) \xrightarrow{\left[\begin{smallmatrix}(s_1,2)\\(\src,2)\end{smallmatrix}\right]}((s_1,\src),2) \xrightarrow{\left[\begin{smallmatrix}(s_2,4)\\(s_2,4)\end{smallmatrix}\right]} 
\\ ((s_2,s_2),4) \biggl ( \xrightarrow{\left[\begin{smallmatrix}(\tgt,5+k)\\(\tgt,5+k)\end{smallmatrix}\right]}((\tgt,\tgt),5+k) \biggr)_{k \geq 0}.\end{multline*}
The new cost of Player~$1$ is $\cost_1(\rho') = 28$ proving that $\sigma'_1$ is not a provitable deviation for Player~$1$ w.r.t. $(\sigma_1,\sigma_2)$. 

Remark that if Player~$2$ does not apply this punishment and moves to $s_4$ at time $2$ whatever the behavior of Player~$1$, playing the deviating strategy $\sigma'_1$ would be a profitable deviation for Player~$1$.
\end{example}

\subsection{Studied problems}
\label{section:studiedProblems}

Given a timed network game~$\timedNetGame$, the~price of anarchy of $\timedNetGame$, denoted by~$\POA_{\timedNetGame}$, is the ratio between the worst social welfare of a Nash equilibrium and the social optimum.
Similarly, the \emph{price of stability} of~$\timedNetGame$, denoted by~$\POS_{\timedNetGame}$, is~the ratio between the best social welfare of a Nash equilibrium and the social optimum.
Those values measure the impact of playing selfishly.
In this paper, we address the following three problems:

\begin{problem}[Constrained social welfare]
\label{pb:constrainedSW}
Given a timed network game $\timedNetGame$ and a threshold $x \in \N$,
is the social optimum in~$\timedNetGame$ less than or equal to~$x$?
\end{problem}

\begin{problem}[Constrained existence of a Nash equilibrium]
\label{pb:constrainedNE}
Given a timed network game $\timedNetGame$ with $\nbp$ players and a
family of linear constraints~$(\phi_q)_{q}$ over $\nbp$ variables,
does there exist a Nash equilibrium~$\sigma$ from $\start$ such that
all constraints $\phi_q(\cost(\InitOutcome{\sigma}{\start}))$ are satisfied?
\end{problem}

\begin{problem}[Constrained price of anarchy and stability]
Given a timed network game $\timedNetGame$ and a threshold $x \in \bbQ$,
is the price of anarchy (resp. of stability)
in~$\timedNetGame$ less than or equal to~$x$?
\end{problem}

\begin{theorem}\label{thm-results}
  The constrained-social-welfare problem is in \PSPACE in the
  symmetric case, and in \EXPSPACE in the asymmetric setting.
  The constrained existence of a Nash equilibrium and constrained
  price of anarchy and stability are in \EXPSPACE in both the
  symmetric and asymmetric cases.
  
\end{theorem}

\section{Existence and computation of Nash equilibria}
\label{section:existComputNE}

We~first notice that, similarly to the untimed case~\cite{Bertrand20},
there are more Nash equilibria in non-blind strategies than when restricting to blind strategies:

\begin{restatable}{proposition}{propinterest}
\label{prop-interest}
There exists a timed network game $\timedNetGame$ that admits a Nash equilibrium~$\sigma$ from $\start$, and whose all blind Nash equilibria~$\tau$ from $\start$ are such that
$\SW(\Outcome{\sigma}_{\start}) < \SW(\Outcome{\tau}_{\start})$.
\end{restatable}

In this section, we~transform the infinite
concurrent game associated with a timed network congestion game into a
finite one, preserving the set of costs of Nash equilibria. We use this
finite game to solve the constrained-Nash-equilibrium
problem, and~then explain how to compute witnessing Nash equilibria.

\subsection{Transformation into an equivalent finite concurrent game}
\label{section:finiteConcurGame}

Let $\timedNetGame=(\nbp, \timedNet, (\src_i,\tgt_i)_{1\leq i\leq
  \nbp})$ be a timed network game, and $\MaxInt$ be the largest
integer appearing in the guards of~$\timedNet$.  In~our assumption of
Remark~\ref{rk-winningtraj} that, for all~$1\leq i\leq\nbp$, there
must exist a trajectory $(\vertex_j,\dat_j)_{j\in \N}$ from~$(\src_i,0)$
to~$(\tgt_i,\delta_i)$ for some date~$\delta_i$, we~can additionally
require that $\delta_i\leq \MaxInt+|\Vertex|$. Indeed, consider such a
winning trajectory with minimal~$\delta_i$: then either $\delta_i\leq
\MaxInt$, or for some~$j_0$, $\dat_{j_0-1}<\MaxInt$ and
$\dat_{j_0}\geq\MaxInt$. Since~$\MaxInt$ is the maximal constant
appearing in the guard of~$\timedNet$, we~can then modify the dates
after $\dat_{j_0}$ so that $\dat_{j_0}=\MaxInt$ and
$\dat_{j_0+k}=\MaxInt+k$, while still satisfying the guards.

Let $\MaxCost = \max_{\vertex \in \Vertex}\weight(\vertex)(\nbp)$ be
the maximum cost per time unit that may occur in the game. With the
arguments above, all players can always ensure a cost of at most
$\MaxTime=\MaxCost\cdot(\MaxInt + |V|)$. Since each time unit costs at least
one cost unit, $\MaxCost\cdot(\MaxInt + |V|)$ is also a bound on the
maximum time within which any player must have reached their target
vertex in any Nash equilibrium: if this were not the case, then that
player would have a profitable deviation.  It~follows:

\begin{lemma}
\label{lemma:boundedLength}
Let $\rho = (s_k,\vec a_k, s'_k)_{k\geq 0}$ be a play in~$\CGame$.
For all $1\leq i \leq\nbp$, if $\cost_i(\rho)$ is finite,
then
there exists a position
$k^* \leq \cost_i(\rho)$
such that $s_{k^*} = (\config_{k^*},\dat_{k^*})$ with
$\config_{k^*}(i)=\tgt_i$ and $\dat_{k^*} \leq
\cost_i(\rho)$.

\looseness=-1
This applies to the particular case where $\rho$ is the outcome of a Nash equilibrium, since in this case, for all $1\leq i \leq n$, $\cost_i(\rho)$ is finite, and bounded by~$\MaxTime$.
\end{lemma}

\paragraph{Definition of the finite concurrent game.}
From the lemma above, when looking for Nash equilibrium, we can unfold
the concurrent game~$\CGame$ into a tree (actually, a~directed acyclic graph)
and prune all subtrees at depth~$\MaxTime$. Formally:
\begin{definition}%
  Let $\timedNetGame$ be a timed network game, and
  $\MaxTime=\MaxCost\cdot(\MaxInt + |V|)$.
  With
  this timed network game, we associate the (finite) concurrent game
  $\FCGame=(\FStates,\FActions,\allowbreak(\FAllowed_i)_{1\leq
    i\leq\nbp},\allowbreak\FUpdate,(\Fcost_i)_{1\leq i\leq \nbp})$ defined as follows:
  \begin{itemize}
  \item $\FStates = \{ (\config,\dat)\in\TConf \mid
    0\leq \dat\leq \MaxTime + 1 \}$
    is the set of time-bounded timed configurations;
  \item $\FActions = \{(\vertex,\dat) \in \Actions \mid d \leq
    \MaxTime + 1 \}\cup \{ (\bot,\MaxTime+1)\}$ is the finite set of actions;
  \item for any~$1\leq i\leq \nbp$ and any state
    $(\config,\dat)\in\FStates$ with $\dat<\MaxTime$, we~have
    action~$(\vertex,\dat')\in\FAllowed_i(\config,\dat)$ if, and
    only~if, $\dat<\dat'\leq \MaxTime+1$ and $(\config(i),\vertex)\in
    E$ and $\dat'\in\guard(\config(i),\vertex)$.
    For~states~$(\config,\dat)$ with $\dat\geq \MaxTime$, we~have
    $\FAllowed_i(\config,\dat) = \{(\bot,\MaxTime+1)\}$.
    We~write $\FMov$ for the set of action vectors.
    An action vector~$\vec a=(a_i)_{1\leq i\leq\nbp}$ is in
    $\FValid(\config,\dat)$ if each $a_i$ is in $\FAllowed_i(\config,\dat)$;
  \item for a state $s=(\config,\dat)$ and a valid action vector $\vec
    a=(\vertex_i,\dat'_i)_{1\leq i\leq\nbp}$, writing
    $m=\min\{\dat'_i\mid 1\leq i\leq\nbp\}$, we~have $\FUpdate(s,\vec
    a)=(\config,\MaxTime+1)$ if $m=\MaxTime+1$, and $\FUpdate(s,\vec
    a)=\Update(s,\vec a)$ otherwise;
  \item in the same way as for (infinite) concurrent games, we let
    $\FTrans = \{ (s,\vec a,s') \in \FStates \times \FMov \times
    \FStates \mid \vec a \in \FValid(s)\text{ and } s'=\FUpdate(s,\vec
    a)\}$ be the set of \emph{transitions} of~$\FCGame$.  Cost
    functions~$\Fcost_i$ are defined on~$\FTrans$
    in the same way
    as for~$\CGame$: we~set
    $\Fcost_i((\config,\dat),\vec a,(\config',\dat'))$ as $ (\dat'-\dat)\cdot \weight(\config(i))(\load_{\config}(\config(i)))$ if~$\dat'\leq\MaxTime$, and
as~$0$ otherwise.
  \end{itemize}
\end{definition}

Notice that $\FStates$ has size doubly-exponential in the size of the
input, since the number of players is given in binary.
We~write~$G^F$ for the graph~$(\FStates,\FTrans)$.  Plays, histories,
costs and strategies are defined
as for infinite concurrent games.

\begin{remark}
  By construction of~$\FCGame$, any play $\rho = (s_k,\vec a_k,
  s'_k)_{k \geq 0}$ in $\FCGame$ ends in a self-loop on some
  configuration~$(\config,\MaxTime+1)$, where the only valid action
  vector~$\vec a$ is such that $a_i=(\bot,\MaxTime+1)$ for all~$1\leq
  i\leq \nbp$.  The~prefix of~$\rho$ before entering this loop
  corresponds to a prefix of a play~$\rho'$ in~$\CGame$, with the
  difference that some actions (with $\dat>\MaxTime$) in~$\CGame$ may
  not be available in~$\FCGame$, but can be modified by
  setting~$\dat=\MaxTime+1$. As~a consequence,
  Lemma~\ref{lemma:boundedLength} also holds for plays in~$\FCGame$.
  
  Also, since the set of actions is finite, and that there is a
  single available action in the terminal self-loop, there is a finite
  number of strategies in~$\FCGame$.
\end{remark}

We~now establish a correspondence between the
Nash equilibria of~$\FCGame$ and~$\CGame$:

\begin{restatable}{theorem}{thmequivNE}
	\label{thm:equivNE}
	Let $\timedNetGame$ be a timed network game and let $\CGame$ and $\FCGame$ be its associated infinite and finite concurrent games.
        Let~$x \in (\N \cup \{+\infty \})^\nbp$. Then 
		there exists a Nash equilibrium $\sigma$ in $\CGame$ with ${\cost(\Outcome{\sigma}_{\start}) = x}$
		if, and only~if, 
		there exists a Nash equilibrium $\tau$ in $\FCGame$ with $\Fcost(\Outcome{\tau}_{\start})=x$
\end{restatable}

We prove this result by establishing simulation relations
between $\CGame$ and~$\FCGame$, satisfying special properties for Proposition~\ref{prop:existenceNESimu} below to apply.
We~let  $\simBT \subseteq \States \times
\FStates$ such that, for any $s_1= (\config_1,\dat_1) \in \States$ and for
any $s_2 = (\config_2,\dat_2) \in \FStates$, $s_1 \simBT s_2$ if, and
only~if, either \emph{(i)}~${\dat_1 =d_2<\MaxTime}$  and 
$\config_1 = \config_2$, or \emph{(ii)}~$d_1 > \MaxTime$ and
$d_2 >\MaxTime$.

We then use this relation (and its inverse) in the following generic
proposition in order to prove  the correspondence between the Nash equilibria
in $\CGame$ and $\FCGame$:
\begin{restatable}{proposition}{propexistenceNESimu}
\label{prop:existenceNESimu}
Let %
$\calG=(\States, \Actions, (\Allowed_i)_{1\leq i\leq \nbp}, \Update,(\cost_i)_{1\leq i\leq \nbp})$
and
$\calG'=(\States', \Actions', (\Allowed'_i)_{1\leq i\leq \nbp}, \Update',
(\cost'_i)_{1\leq i\leq \nbp})$
be two $\nbp$-player concurrent games,
$s_0\in\States$ and~$s'_0\in\States'$,
and
$\genSim \subseteq \States \times \States'$ be a relation such that
\begin{enumerate}
\item $s_0 \genSim s'_0$;\label{item:simInit}
\item \label{item-cond2}
  there exists $\lambda \in \N$ such that for any NE~$\sigma$ in
  $\calG$ from~$s_0$, for any $1\leq i \leq\nbp$, it~holds
  $\cost_i(\Outcome{\sigma}_{s_0}) \leq \lambda$;
\item for all plays $\rho \in \Plays_{\calG}$ and
  $\rho' \in \Plays_{\mathcal{G}'}$ such that $\rho \genSim \rho'$
  (\IE, $\rho_j\genSim\rho'_j$ for all $j\in\bbN$),
  all $1\leq i\leq\nbp$, if $\cost_i(\rho) \leq \lambda$
  or $\cost'_i(\rho') \leq \lambda$,
  then $\cost_i(\rho) = \cost'_i(\rho')$;\label{item:simCost}
\item For all $s \in \States$, for all $\vec a \in \Valid(s)$,
  for all $s' \in \States'$,
  if $s \genSim s'$, then there exists $\vec a' \in \Valid'(s')$ such that:
\begin{enumerate}
\item $\Update(s,\vec a) \genSim \Update'(s',\vec a')$; \label{item:sim1}
\item   \label{item:sim2} for all $1\leq i\leq\nbp$,
  for all $b'_i \in \Allowed'_i(s')$, there exists $b_i \in \Allowed_i(s)$
  such that
  \[
  \Update(s,(a_{-i},b_i)) \genSim \Update'(s',(a'_{-i},b'_i)).
  \]
\end{enumerate}
\end{enumerate}

Then for any Nash equilibrium~$\sigma$ in~$\calG$ from~$s_0$,
there exists a Nash equilibrium~$\sigma'$ in~$\calG'$ from~$s'_0$
such that $\cost(\Outcome{\sigma}_{s_0})=\cost'(\Outcome{\sigma'}_{s'_0})$.
\end{restatable}

\subsection{Existence of Nash equilibria}\label{sec-ExistNE}

\begin{restatable}{theorem}{thmexistenceNE}
	\label{thm:existenceNE}
	Let $\timedNetGame$ be a timed network game and let $\CGame$
        be its associated concurrent game. There exists a Nash
        equilibrium~$\sigma$ from~$\start$ in $\CGame$.
\end{restatable}

We~prove this result using a \emph{potential
function}~$\Psi$~\cite{MS-geb96}; a~potential function is a function
that assigns a non-negative real value to each strategy profile, and
decreases when the profile is \emph{improved} (in~a sense that we
explain below). Since those improvements are performed among a finite
set of strategy profiles, this entails convergence of the sequence of
improvements.

An \emph{improvement} consists in changing the
strategy of one of the players by a better strategy for that player
(if~any), in the sense that their individual cost descreases. Because
the game admits a potential function over finitely many strategies,
this \emph{best-response dynamics} must converge in finitely many
steps, and the limit is a strategy profile where no single player can
improve their strategy, \IE a Nash equilibrium.

We~prove that timed network games admit a potential function when
considering winning blind strategies; hence there exists Nash equilibria
for that set of strategies.
We~then prove that a Nash equilibrium in this
restricted setting remains a Nash equilibrium
w.r.t. the set of all strategies.

\smallbreak
For all $1\leq i\leq\nbp$, there are a finite number of
strategies of Player~$i$ in~$\FCGame$.
It~follows that the set of winning blind strategies of Player~$i$ is
also finite; we~let
\[
\Sigma^{\WB}_i = \{\sigma_i \mid
\sigma_i \text{ is a winning blind strategy of Player~$i$}\}.
  \]
Let $\Sigma^{\WB} = \Sigma^{\WB}_1 \times \ldots \times \Sigma^{\WB}_n$;
$\Sigma^{\WB}$ is a finite set of strategy profiles. We~consider a restriction
of~$\FCGame$ in which for all $1\leq i\leq \nbp$, Player~$i$ is
only allowed to play a strategy in~$\Sigma^{\WB}_i$.
A~$\Sigma^{\WB}$-Nash equilibrium
then is a strategy profile in~$\Sigma^{\WB}$
such that for all $1\leq i\leq\nbp$, Player~$i$ has no profitable deviation
\emph{in~$\Sigma^{\WB}_i$}. 
The~next proposition implies that $\FCGame$ always admits
$\Sigma^{\WB}$-Nash equilibria~\cite{MS-geb96}.
\begin{restatable}{proposition}{proppotentialGame}
	\label{prop:potentialGame}
	Let $\timedNetGame$ be a timed network game and let
        $\FCGame$ be its associated finite concurrent
        game and $\Sigma^{\WB}$  be the set of winning blind strategy profiles.
        The~game
        $\FCGame$ restricted to  strategy profiles in~$\Sigma^{\WB}$
        has a potential function.
\end{restatable}

Moreover
each $\Sigma^{\WB}$-Nash equilibrium corresponds to a Nash equilibrium with the same cost profile (Proposition~\ref{prop:ENrestToEN}). 
Notice that the converse result fails to hold, as proved in Proposition~\ref{prop-interest}.
\begin{restatable}{proposition}{propENrestToEN}
  \label{prop:ENrestToEN}
  Let $\timedNetGame$ be a timed network game. Let~$\FCGame$
  be its associated finite concurrent game, and $\Sigma^{\WB}$ the set
  of winning blind strategy profiles.
  If~there exists a $\Sigma^{\WB}$-Nash equilibrium~$\sigma$ in
  $\FCGame$ from~$\start$, then there exists
  a Nash equilibrium $\tau$ from~$\start$ in $\FCGame$ with the same costs for all players.
  \end{restatable}

\subsection{Computation of Nash equilibria}\label{sec-computeNE}

\paragraph{Characterization of outcomes of Nash equilibria.}

In this section, we~develop an algorithm for deciding (and computing) the existence of
 a Nash equilibrium satisfying a given constraint on
the costs of the outcome. We~do this by computing the maximal
punishment that a coalition can inflict to a deviating player. This
value can be computed using classical techniques in a two-player
zero-sum concurrent game: from a timed configuration~$s=(\config,\dat)$,
this value is defined as
\[
\LowVal_i(s) = \sup_{\sigma_{-i} \in \Sigma_{-i}(s)} \;
\infp_{\sigma_{i}\in\Sigma_i(s)} \cost_i (\InitOutcome{\sigma_{-i},\sigma_i}{s}).
\]

For all histories $h = (s_k,\vec a_k, s'_k)_{0\leq k < \ell}$ in
$\HistFG(s)$, we~let $\Visit(h)$ be the set of players who visit their
target vertex along $h$. Formally, for the empty history
from~$s=(\config,\dat)$, we let $\Visit(h)=\{ 1\leq i \leq\nbp \mid
  \config(i)=\tgt_i\}$. If~$h = (s_k,\vec a_k, s'_k)_{0\leq k <
    \ell}$ is non-empty, writing~$s_k=(\config_k,\dat_k)$ for
  all~$0\leq k\leq\ell$, we~have $\Visit(h) = \{ 1\leq i \leq\nbp \mid
  \exists 1 \leq k \leq \ell.\ \config_k(i) = \tgt_i\}$.

The following theorem is a characterization of outcomes of Nash
equilibria.
Similar characterizations were proven in~\cite{Bertrand20}
for (untimed) network congestion games, and in~\cite{Klimos12,Almagor18}
for generic concurrent games:

\begin{restatable}{theorem}{thmcritNE}
\label{thm:critNE}
Let  $\timedNetGame$ be a timed network congestion game and
$\FCGame$ be its associated finite concurrent game.
A~play $\rho=(s_k,\vec a_k,s'_k)_{k\in\bbN} \in \PlaysFG(\start)$
is the outcome of a Nash equilibrium from~$\start$ in $\FCGame$ if,
and only~if,
\begin{multline}
\forall 1\leq i\leq\nbp.\  \forall k \in \N.\
\forall b_i \in \FAllowed_i(s_k).\
i \not\in \Visit(\rho_{<k}) \implies \label{eq-NE}\\
\Fcost_i(\rho_{\geq k}) \leq \LowVal_i(s') + \Fcost_i(s_k, (\vec a_{k,-i}, b_i),s')
\end{multline}
where $s' = \FUpdate(s_k, (\vec a_{k,-i}, b_i))$.
\end{restatable}

The values $\LowVal_i(s)$ can be computed by transforming the finite
game~$\FCGame$ into a two-player game, since the deviating player
competes against the coalition of all the other players.  In~this
transformation, we do not need to keep track of the position of all
the other players individually, since their aim now only is to maximize the cost
for the deviating player. This two-player game thus has size exponential,
and each $\LowVal_i(s)$ can be computed in exponential time. Then:
\begin{restatable}{proposition}{propprobTWOEXPSAPCE}
  Problem~\ref{pb:constrainedNE} can be decided in \EXPSPACE, both for the symmetric and for the asymmetric cases.
\end{restatable}

\begin{proof}
  Let $\timedNetGame$ be a timed network game
  and $(\phi_q)_q$ be a set of linear constraints.
  We prove that we
can decide in exponential space the existence of a Nash
equilibrium~$\sigma$ in $\FCGame$ such that
$\Fcost(\InitOutcome{\sigma}{\start})$
satisfies all linear constraints~$(\phi_q)_q$.
By~Theorem~\ref{thm:equivNE}, this entails the same result in $\CGame$.

As already argued, each play~$\rho$ in $\FCGame$ ends up in a loop
after at most $\MaxTime$ steps. Our algorithm will guess such a play
until the loop; the~play can be stored in exponential space. The~algorithm
will then
check that it is a valid play, that Eq.~\ref{eq-NE} holds at each
step (which requires to compute $\LowVal$),
compute the costs paid by the players until they reach their
targets, and check that the cost constraints~$(\phi_q)_q$ are
satisfied.

This algorithm uses pseudo-polynomial space, as it has to store a
polynomial quantity of data for each player. Notice that the algorithm
would not store all $\LowVal$ values as the number of values would be
doubly-exponential; instead, it will re-compute those values on-demand.
\end{proof}

\begin{remark}
Our definition of plays and histories include the action vectors at
each step, which implies that strategies observe the actions of all
players and can base their decisions on that information. In~our
setting, however, it~is possible to identify one of the deviating
players based only on the configurations; in~particular, if a single
player deviates, they can be identified and punished. As~a
consequence, our~results still hold in the setting where strategies
may only depend on the sequence of timed configurations.
\end{remark}

\section{Social optimality and prices of anarchy and stability}
\label{section:socOptiPriceAnarch}
In this section, given a timed network game $\timedNetGame$, we study
the social optimum~$\SocOpti_{\timedNetGame}$. In~order to obtain the
best social welfare, the~players aim at minimizing the sum of their
costs whatever their selfish interrest. Thus we want to find a play
from~$\start$ in the concurrent game~$\CGame$ such that the sum of the
costs of all the players is as small as possible. 
We~also consider the price of stability $\POS_{\timedNetGame}$ (resp. of anarchy $\POA_{\timedNetGame}$) to know how far is the social optimum from the best (resp. worst) social welfare of a Nash equilibrium from $\start$ in $\timedNetGame$.

Before studying the social optimum and the prices of anarchy and stability, we explain how to solve the constrained-social-welfare problem both in asymmetric and symmetric timed network games. 
First notice that the following lemma is a consequence of Lemma~\ref{lemma:boundedLength}:
\begin{lemma}
\label{lem:boundedTimeSW}
For all plays $\rho \in \PlaysG(\start)$ such that $\SW(\rho)$ is finite,
for all ${i \in \Players}$, there exists $k_i \leq \SW(\rho)$ such that
$\rho_{k_i}=(\config_{k_i}, \dat_{k_i})$ with 
\emph{(i)}~$\config_{k_i}(i) = \tgt_i$ and
\emph{(ii)}~$\dat_{k_i} \leq \SW(\rho)$.
\end{lemma}

\subsection{Constrained-social-welfare problem: asymmetric case}
First of all, let us assume that the players' objectives are
asymmetric. In~this setting, Problem~\ref{pb:constrainedSW}
can be solved by
non-deterministically guessing a (finite) play in~$\CGame$ step-by-step:
Lemma~\ref{lem:boundedTimeSW} gives a polynomial bound on the size of the configurations to be guessed;
keeping track of the set of players who reached their targets
requires exponential space. By Savitch's theorem, we~get:

\begin{restatable}{proposition}{propconstSWEXPSPACEasym}
\label{prop:constSWEXPSPACEasym}
The constrained-social-welfare problem (Problem~\ref{pb:constrainedSW})
can be decided in \EXPSPACE if the players' objectives are asymmetric.
\end{restatable}

\subsection{Constrained-social-welfare problem: symmetric case}

For symmetric TNGs, the objectives of the players are identical: there exist
$\src$ and $\tgt$ in~$\Vertex$ such that $\src_i=\src$ and $\tgt_i=\tgt$ for all~$1\leq i\leq\nbp$.

We could of course reuse the algorithm we developed for the asymmetric case, resulting in an \EXPSPACE algorithm. 
Nevertheless, we can refine this approach by considering a weighted graph in which we only take into account \emph{abstract timed configurations}. An~abstract timed configuration $\tilde{\config}$ is a tuple  $(P_A,P_W,d) \in [0,\nbp]^{\Vertex} \times [0,\nbp]^{\Vertex} \times\bbN$ where \emph{(i)}~$P_A$~maps each vertex to the number of \emph{active players} (players who have not visited the target vertex yet) in that vertex; \emph{(ii)}~$P_W$~maps each vertex to the number of \emph{ winning players} (players who have already visited their target set) in that vertex and \emph{(iii)}~$d$~is the current time.

Abstract timed configurations store enough information to compute the social welfare of a play in symmetric TNGs and give rise to a weighted graph $\mathcal{W} = (A,B, \tilde{w})$ (see App.~\ref{ann-section4} for a formal definition).
The set $A$ is the set of abstract timed configurations, $B \subseteq A \times A$ is the set of edges such that there exists an edge $(\tilde{\config_1},\tilde{\config_2}) \in B$ between two abstract timed configurations $\tilde{\config_1}$ and $\tilde{\config_2}$ if there exists a valid action for each player regarding their position given by $\tilde{\config_1}$ such that updating  $P_A$, $P_B$ and the current time w.r.t. this action vector leads to the abstract configuration $\tilde{\config_2}$. The weight function $\tilde{w}: B \rightarrow \N$ represents the sum of the costs of active players for an edge. Notice that the winning players are taken into account to compute the cost of an active player since their presence in a vertex influences the load in that vertex.

A path~$p$ in $\calW$ is a finite sequence of abstract timed configurations consistent with
the graph structure $(A,B)$, starting from the initial vertex
$(\tilde{\start}; \{0\}^\Vertex; 0)$ (assuming $\src\not=\tgt$) with $\tilde{\start}: V \rightarrow [0,\nbp]: v \mapsto \#\{1 \leq i \leq n \mid \InitConfig(i) = v \}$.
The~cost of a path~$p$ in $\calW$ is
either the sum of the weights $\tilde{w}(a_1,a_2)$ along the path
until visiting a final vertex (where $P_A(\vertex)=0$ for all
$\vertex\in\Vertex$), or $+\infty$ if no such vertices appear
along~$p$.
Clearly enough, the abstract weighted graph encodes the
trajectories of~$\timedNetGame$ in the following sense:
\begin{lemma}
\label{lemma:abstractPath}
Let $\timedNetGame$ be a timed network game and $\mathcal{W}$ be its associated abstract weighted graph.
For all $c \in \N$, there exists a play $\rho \in \Plays_\CGame(\start)$  such that $\SW(\rho) = c$ if, and only~if, there exists a path $p$ in $\mathcal{W}$ with cost $c$.
\end{lemma}

The constrained-social-welfare problem for symmetric objectives can
then be solved by non-deterministically guessing the successive
vertices of a path~$p$ in~$\mathcal{W}$,
step-by-step. The~constraint~$c$ gives a bound on the length of the
path, so that the algorithm runs in polynomial space.

\begin{restatable}{proposition}{propconstrainedSWsym}
  \label{prop:constrainedSWsym}
The constrained-social-welfare problem
(Problem~\ref{pb:constrainedSW}) can be decided in \PSPACE if the
players' objectives are symmetric.
\end{restatable}

\subsection{Social Optimum and  Prices of Anarchy and Stability}

\paragraph{Optimum Social.}
We now explain how we can compute the exact social optimum: noticing
that the social optimum can be bounded by $\nbp\cdot\MaxTime$, it can
be computed by performing a binary search, iteratively applying the
algorithm above. Computing the social optimum can thus be performed in
polynomial space in the symmetric case, and exponential space in the
asymmetric case.

\paragraph{Prices of Anarchy and Stability.}
The constrained-price-of-anarchy (resp. stability) problem can now be solved using our
algorithms for solving the constrained-social-welfare 
and constrained-Nash-equilibrium problems: thanks to the pseudo-polynomial bound
$\nbp\cdot\MaxTime$ on the social welfare of the social optimum
and on the social welfare of any Nash equilibrium,
and the fact that those values are integers,
we~can perform binary searches for the
exact social welfare of the social optimum
and for the worst (resp. best) social welfare of a Nash
equilibrium.  This only requires a polynomial number of iterations,
so that the whole algorithm runs in exponential space.

\newcommand{\etalchar}[1]{$^{#1}$}

\clearpage
\appendix

\section{Proofs of Section~\ref{sec-prelim}}

\lemmablindpath*

\begin{proof}
For a contradiction, assume that this were not the case, and consider
two strategy profiles $(\sigma_{-i},\sigma_i)$ and
$(\sigma'_{-i},\sigma_i)$, and the first position~$j$ where
$\traj_i(\Outcome{\sigma_{-i},\sigma_i}_{s_0})$ and
$\traj_i(\Outcome{\sigma'_{-i},\sigma_i}_{s_0})$ differ. This position
corresponds to two (possibly different) positions~$k=\mu_i(j)$
and~$k'=\mu'_i(j)$ along $\Outcome{\sigma_{-i},\sigma_i}_{s_0}$
and $\Outcome{\sigma'_{-i},\sigma_i}_{s_0}$. 

The prefixes~$h$ and~$h'$ of
$\Outcome{\sigma_{-i},\sigma_i}_{s_0}$ and
$\Outcome{\sigma'_{-i},\sigma_i}_{s_0}$ up to positions~$k-1$
and $k'-1$ are projected on the same trajectory for Player~$i$, so that
$\sigma_i$ returns the same action for both histories. Moreover, by construction
of~$k$ and~$k'$, the~delay proposed in this action must be shorter than
(or~equal~to) the delays proposed by the strategies in~$\sigma_{-i}$
and~$\sigma'_{-i}$ after histories~$h$ and~$h'$;
then the same edge is applied for Player~$i$
from position~$k-1$ of $\Outcome{\sigma_{-i},\sigma_i}_{s_0}$
and position~$k'-1$ of $\Outcome{\sigma'_{-i},\sigma_i}_{s_0}$,
contradicting the fact that
$\traj_i(\Outcome{\sigma_{-i},\sigma_i}_{s_0})$ and
$\traj_i(\Outcome{\sigma'_{-i},\sigma_i}_{s_0})$ differ at
position~$j$.
\end{proof}

\section{Proofs of Section~\ref{section:existComputNE}}
\subsection{Details of the proof of Proposition~\ref{prop-interest}}

\propinterest*

\begin{proof}
We consider the timed network game $\timedNetGame = (3, \timedNet, (\src,\tgt))$ such that its timed network is given in Figure~\ref{ex:lowerSW}. 
First of all, we focus on strategy profiles composed only by blind strategies.
\begin{figure}[ht]
	\centering
	\scalebox{0.85}{\begin{tikzpicture}
		\node[draw, rounded corners=5pt] (s0) at (0,0){\nodeCong{1cm}{$\src$}{$x \mapsto 1$}};
		\node[draw, rounded corners=5pt] (s1) at (2,1.5){\nodeCong{1.5cm}{$s_1$}{$x \mapsto 2x$}};
		\node[draw, rounded corners=5pt] (s2) at (4,1.5){\nodeCong{1cm}{$s_2$}{$x\mapsto 1$}};
		\node[draw, rounded corners=5pt] (s3) at (6,1.5){\nodeCong{1cm}{$s_3$}{$x \mapsto 2$}};
		
		\node[draw, rounded corners=5pt] (s4) at (8,1.5){\nodeCong{1cm}{$s_4$}{$x \mapsto 1$}};
		
		\node[draw, rounded corners=5pt] (s5) at (10,1.5){\nodeCong{1cm}{$s_5$}{$x\mapsto 2$}};
		\node[draw, rounded corners=5pt] (s6) at (12,1.5){\nodeCong{1.5cm}{$s_6$}{$ x \mapsto 2x$}};
		
		\node[draw, rounded corners=5pt] (s7) at (2,-1.5){\nodeCong{1.5cm}{$s_7$}{$x \mapsto 3x$}};
		\node[draw, rounded corners=5pt] (s8) at (4,-1.5){\nodeCong{1cm}{$s_8$}{$x\mapsto x$}};
		\node[draw, rounded corners=5pt] (s9) at (6,-1.5){\nodeCong{1cm}{$s_9$}{$x \mapsto x$}};
		\node[draw, rounded corners=5pt] (s10) at (8,-1.5){\nodeCong{1cm}{$s_{10}$}{$x \mapsto x$}};
		\node[draw, rounded corners=5pt] (s11) at (10,-1.5){\nodeCong{1cm}{$s_{11}$}{$x \mapsto x$}};
		\node[draw, rounded corners=5pt] (s12) at (12,-1.5){\nodeCong{1cm}{$s_{12}$}{$x \mapsto x$}};
		
		\node[draw, rounded corners=5pt] (s13) at (14,0){\nodeCong{1cm}{$\tgt$}{$x\mapsto 1$}};
		
		\node[draw, rounded corners=5pt] (s14) at (4,0){\nodeCong{1cm}{$s_{14}$}{$x \mapsto x$}};
		
		\node[draw, rounded corners=5pt] (s15) at (8,0){\nodeCong{1cm}{$s_{15}$}{$x \mapsto 4$}};

		\draw[->] (s0) to node[left]{$[1]$} (s1);
		\draw[->] (s1) to node[above]{$[2]$} (s2);
		\draw[->] (s2) to node[above]{$[3]$} (s3);
		\draw[->] (s3) to node[above]{$[4]$} (s4);
		\draw[->] (s4) to node[above]{$[5]$} (s5);
		\draw[->] (s5) to node[above]{$[6]$} (s6);
		\draw[->] (s6) to node[above]{$[7]$}(s13);
		\draw[->] (s0) to node[above]{$[1]$} (s7);
		\draw[->] (s7) to node[above]{$[2]$} (s8);
		\draw[->] (s8) to node[above]{$[3]$} (s9);
		\draw[->] (s9) to node[above]{$[4]$} (s10);
		\draw[->] (s10) to node[above]{$[5]$} (s11);
		\draw[->] (s11) to node[above]{$[6]$} (s12);
		\draw[->] (s12) to node[above]{$[7]$} (s13);
		
		\draw[->] (s13) to [loop right] (s13);
		
		\draw[->] (s1) to node[left]{$[2]$} (s14);
		\draw[->] (s14)to node[left ]{$[3]$} (s9);
		
		\draw[->](s3) to node[left]{$[4]$} (s15);
		 \draw[->] (s15) to node[left]{$[5]$} (s11);		
		
	\end{tikzpicture}
	}
	\caption{Timed network game with a lower social welfare of an NE with non-blind stragies}
	\label{ex:lowerSW}
\end{figure}
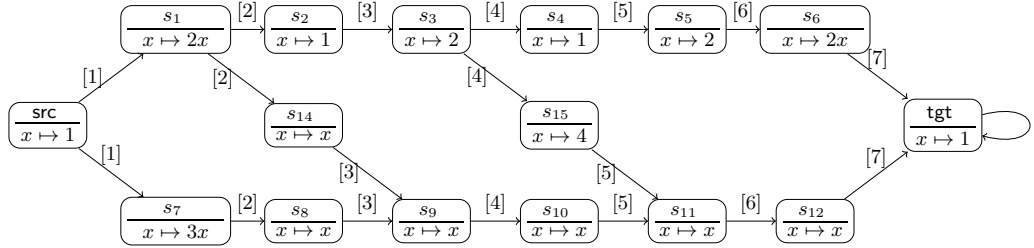

There are four different trajectories of blind strategies:

\begin{enumerate}
	\item $\pi_1 = (\src,0)(s_1,1)(s_2,2)(s_3,3)(s_4,4)(s_5,5)(s_6,6)(\tgt,7+k)_{k\geq 0}$
	\item $\pi_2 = (\src,0)(s_7,1)(s_8,2)(s_9,3)(s_{10},4)(s_{11},5)(s_{12},6)(\tgt,7+k)_{k\geq 0}$
	\item $\pi_3=(\src,0)(s_1,1)(s_{14},2)(s_9,3)(s_{10},4)(s_{11},5)(s_{12},6)(\tgt, 7+k)_{k\geq 0}$
	\item $\pi_4= (\src,0)(s_1,1)(s_2,2)(s_3,3)(s_{15},4)(s_{11},5)(s_{12},6)(\tgt,7+k)_{k\geq 0}$	
\end{enumerate}

By abuse of notation, in the rest of this example we identify a trajectory $\pi$ with a blind strategy $\sigma$ such that $\Outcome{\sigma}_{\start} = \pi$. 

 We consider the social welfare of the following strategy profiles:

\begin{minipage}{0.45\linewidth}
\begin{itemize}
	\item $(\pi_1,\pi_1,\pi_1) \leadsto 19 \cdot 3 = 57 $ 
	\item $(\pi_2,\pi_2,\pi_2) \leadsto 25 \cdot 3 = 75 $ 
	\item $(\pi_3,\pi_3,\pi_3) \leadsto 22 \cdot 3 = 66 $ 
	\item $(\pi_4,\pi_4,\pi_4) \leadsto 20 \cdot 3 = 60$ 
	
	\item[]
	
	\item $(\pi_1,\pi_1,\pi_2) \leadsto 15 \cdot 2 + 9 = \mathbf{39}$ 
	\item $(\pi_1,\pi_1,\pi_3) \leadsto 17 \cdot 2 + 12 = 46$ 
	\item $(\pi_1,\pi_1,\pi_4) \leadsto 17 \cdot 2 + 16 = 50 $
	\item[]
	\item $(\pi_2,\pi_2,\pi_1) \leadsto 17 \cdot 2 + 11 = 45$ 
	\item $(\pi_2,\pi_2,\pi_3) \leadsto 21 \cdot 2 + 16 = 58$ 
	\item $(\pi_2,\pi_2,\pi_4) \leadsto 19 \cdot 2 + 16 = 54$ 	
\end{itemize}
\end{minipage}\hfill
\begin{minipage}{0.45\linewidth}
	\begin{itemize}
		\item $(\pi_3,\pi_3,\pi_1) \leadsto  17 \cdot 2 + 15 = 49$ 
		\item $(\pi_3,\pi_3,\pi_2) \leadsto19 \cdot 2 + 17 = 55$ 
		\item $(\pi_3,\pi_3,\pi_4) \leadsto 19 \cdot 2 + 20 = 58$ 
		\item[]
		\item $(\pi_4,\pi_4,\pi_1) \leadsto 18 \cdot 2 + 15 = 51$ 
		\item $(\pi_4,\pi_4,\pi_2) \leadsto 18 \cdot 2 + 13 = 49$ 
		\item $(\pi_4,\pi_4,\pi_3) \leadsto 20 \cdot 2 + 16 = 56$ 
		\item[]
		\item $(\pi_1,\pi_2,\pi_3) \leadsto 13 + 13 + 14 = 40$
		\item $(\pi_1,\pi_2,\rho_4) \leadsto 13 + 11 + 16 = 40$ 
	 	\item $(\pi_1,\pi_3,\pi_4) \leadsto 15 + 14 + 18 = 47$ 	
		\item $(\pi_2,\pi_3,\pi_4) \leadsto 15 + 16 + 18 = 49$ 
	\end{itemize}
\end{minipage}

We prove that the social welfare of all Nash equilibria with blind strategies is greater than 39.
All social welfare of those blind strategy profiles are greater than 39 except $(\pi_1,\pi_1, \pi_2)$ but this is not a Nash equilibrium since with this strategy profile Player~$2$ has a cost of $15$ and he has an incentive to deviate from $\pi_1$ to $\pi_3$ and obtains a cost of $14$. \\

We now prove that there exists a Nash equilibrium from $\start$  with social welfare equal to $39$ and such that one player uses a non-blind strategy.
We consider the strategy profile $\sigma = (\sigma_1,\sigma_2,\sigma_3)$ where, roughly speaking,
Player $1$ and Player $2$ follow the trajectory $ \pi_1$ and
Player~$3$ follows the trajectory $\pi_2$.
If Player 3-$i$, for $i \in \{1,2\}$, deviates in $\src$ then Player~$i$ follows trajectory $\pi_3$.
 If Player 3-$i$, for $i \in \{1,2\}$, deviates in $s_1$ then Player~$i$ follows trajectory $\pi_4$. 

The outcome of $\sigma$ from $\start$ is: 

\begin{align*}&\start \xrightarrow{ \begin{bmatrix} (s_1,1) \\ (s_1,1) \\ (s_7,1) \end{bmatrix}} ((s_1,s_1,s_7),1) 
\xrightarrow{\begin{bmatrix} (s_2,2) \\ (s_2,2) \\ (s_8,2)\end{bmatrix}} ((s_2,s_2,s_8),2)
\xrightarrow{ \begin{bmatrix} (s_3,3)\\ (s_3,3) \\ (s_9,3)\end{bmatrix}} ((s_3,s_3,s_9),3)
\xrightarrow{ \begin{bmatrix} (s_4,4)\\ (s_4,4) \\(s_{10},4) \end{bmatrix}}\\
&((s_4,s_4,s_{10}),4)
\xrightarrow{\begin{bmatrix} (s_5,5)\\(s_5,5)\\(s_{11},5) \end{bmatrix}} ((s_5,s_5,s_{11}),5)
\xrightarrow{\begin{bmatrix} (s_6,6)\\ (s_6,6)\\ (s_{12},6)\end{bmatrix}} ((s_6,s_6,s_{12}),6)\\
&\left(\xrightarrow{\begin{bmatrix} (\tgt, 7 + k)\\ (\tgt,7+k)\\ (\tgt,7+k)\end{bmatrix}}((\tgt,\tgt,\tgt), 7+k)\right)_{k\geq0}.\end{align*}

with cost $(15,15,9)$ thus the social welfare is equal to $39$.

We prove that $(\sigma_1,\sigma_2,\sigma_3)$ is a Nash equilibrium.

For Player~$i$ with $i \in \{1,2\}$:
\begin{itemize}
\item if he deviates in $v_0$ then the cost of Player~$i$ is $21$;
\item if he deviates in $v_1$ then the cost of Player~$i$ is equal to $16$;
\item if he deviates in $v_3$ then the cost of Player~$i$ is equal to $16$.
\end{itemize}

For Player~$3$:
\begin{itemize}
\item if he deviates and follows trajectory $\pi_1$, then the cost of Player~$3$ is $19$;
\item if he deviates and follows trajectory $\pi_3$, then the cost of Player~$3$ is $12$;
\item If he deviates and follows trajectory $\pi_4$, then the cost of Player~$3$ is $16$.
\end{itemize}

That proves that there is no profitable deviation for any player. 
\end{proof}

\subsection{Proofs of Section~\ref{section:finiteConcurGame}}

\thmequivNE*

\paragraph{Simulations between the graph structures.}
In order to prove Theorem~\ref{thm:equivNE}, we have to define how to
simulate the actions of the players from~$\CGame$ to~$\FCGame$ and
vice-versa.

We first define an equivalence relation $\simBT \subseteq \States \times
\States$ such that, for any $s_1= (\config_1,\dat_1) \in \States$ and for
any $s_2 = (\config_2,\dat_2) \in \States$, $s_1 \simBT s_2$ if, and
only~if, either \emph{(i)}~${\dat_1 =d_2<\MaxTime}$  and 
$\config_1 = \config_2$, or \emph{(ii)}~$d_1 > \MaxTime$ and
$d_2 >\MaxTime$.

Notice that since $\FStates \subseteq \States$, the~relation~$\simBT$
is also well-defined on $\States \times \FStates$.

\begin{proposition}[Action vector simulation from $\CGame$ to $\FCGame$]
\label{prop:simInfToFin}
For all $s_1 \in \States$ and $s'_1 \in \FStates$ such that $s_1 \simBT s'_1$,
for all $\vec a \in \Valid(s_1)$,
there exists $\vec a' \in \FValid(s'_1)$ such that 
\begin{enumerate}
\item $\Update(s_1,\vec a) \simBT \FUpdate(s'_1,\vec a')$;
\item for all $i \in \Players$, for all $b'_i \in \FAllowed_i(s'_1)$, there exists $b_i \in \Allowed_i(s_1)$ such that
  \[
  \Update(s_1,(a_{-i},b_i)) \simBT \FUpdate(s'_1,(a'_{-i},b'_i)).
  \]
\end{enumerate}
\end{proposition}

\begin{proof}
  Let $s_1 = (\config_1,\dat_1) \in \States$,
  let $\vec a= (\vertex_i, \datt_i)_{1\leq i\leq \nbp} \in \Valid(s_1)$, and
  let $s'_1 = (\config'_1,\dat'_1) \in \FStates$ such that $s_1 \simBT s'_1$.
  Let $s_2 = (\config_2,\dat_2) = \Update(s_1,\vec a)$.

  We have to define $\vec a' = (\vertex'_i, \datt'_i)_{1\leq i\leq\nbp}
  \in \FValid(s'_1)$ such that, letting $s'_2=\FUpdate(s'_1,\vec a')$,
  it~holds $s_2 \simBT s'_2$.
  This depends on~$\dat_1$:
\begin{enumerate}

\item \textbf{if $\dat_1 < \MaxTime$:} since $s_1 \simBT s'_1$,
  we have $s_1=s'_1$.
  We let $a'_{i} = (\vertex_i,\min\{\datt_i,\MaxTime + 1\})$
  for all $1\leq i\leq \nbp$.
  Notice that $\vec a' \in \FValid(s'_1)$.
  Indeed, for all $1\leq i\leq \nbp$,
  if $\min\{\datt_i, \MaxTime +1 \} = \datt_i$,
  since $(\vertex_i,\datt_i) \in \Allowed_i(s_1)$, then also
  $(\vertex_i,\datt_i) \in \FAllowed_i(s'_1)$ (as~${s_1 = s'_1}$);
  otherwise, if $\min\{\datt_i, \MaxTime +1 \} = \MaxTime +1$,
  then $t_i \geq \MaxTime +1 \geq \MaxInt$, and since $(\vertex_i,\datt_i)\in
  \Allowed_i(s_1)$, then also $(\vertex_i,\MaxTime + 1) \in \FAllowed_i(s'_1)$
  (as~${s_1 = s'_1}$, and $\MaxInt$ is the largest integer constant appearing
  in the guards).

  We let $s'_2 = (\config'_2,\dat'_2) = \FUpdate(s'_1,\vec a')$,
  $\datt^* = \min \{ \datt_i \mid 1\leq i\leq\nbp\}$, and
  $\datt'^* = \min \{\datt'_i \mid 1\leq i\leq \nbp\}$.
  Remember that $\Select(\vec a) = \{1\leq i\leq \nbp \mid \datt_i=\datt^*\}$.
  
\begin{enumerate}
\item If $\datt^* \geq \MaxTime +1$: we then have 
  $\dat_2 =\datt^* \geq \MaxTime +1$, hence $\datt'^* = \MaxTime+1$ and
  $\dat'_2 = \MaxTime +1$. It~follows that $s_2 \simBT s'_2$.
\item otherwise, $\datt^* = \datt'^*$, and it follows that
  $\Select(\vec a) = \Select(\vec a')$. Thus for all $i \in \Select(\vec a)$,
  we~have $a_i = a'_i$. It~follows
  $\Update(s_1,\vec a) =
  \FUpdate(s'_1,\vec a')$, and $s_2 \simBT s'_2$.
\end{enumerate}

That proves the first assertion. We now prove the second assertion. 
Take $1\leq i\leq \nbp$, and $b'_i \in \FAllowed_i(s'_1)$.  Since
$\dat'_1 < \MaxTime$, we~have $b'_i = (\vertex'_i,\datT'_i)$ for some
$\datT'_i \leq \MaxTime +1$ and some $\vertex'_i \in \Vertex$.
We~let $b_i = b'_i$.
Since $s_1= s'_1$, we~have $b_i \in \Allowed_i(s_1)$.
\begin{enumerate}
\item If $\datT'_i \leq \MaxTime$, then $\min (\{ \datt'_j \mid j \neq i \}
  \cup \{\datT'_i\}) = \min (\{\datt_j \mid j \neq i \} \cup \{\datT'_i\})$
  since $\datt'_j=\datt_j$ for
  all $\datt'_j \leq \MaxTime$. Thus $\Select(a_{-i},b_i)
  = \Select(a'_{-i},b'_i)$ and since $\vertex'_j = \vertex_j$ for all $j\neq
  i$, we~get $\Update(s_1,(a_{-i}, b_i)) = \FUpdate(s'_1,(a'_{-i}, b'_i))$.
\item otherwise, $\datT'_i = \MaxTime + 1$;
  in~that case, if there exists $j \neq i$ such that
  $\datt'_j \leq \MaxTime$, then $\datt_j = \datt'_j$, and $i \notin
  \Select(a'_{-i}, b'_i)$ and $i \notin \Select(a_{-i},b_i)$. It
  follows that $\Update(s_1,\vec a) = \Update(s_1,(a_{-i}, b_i))$, and
  $\FUpdate(s'_1,\vec a') = \FUpdate(s'_1,(a'_{-i}, b'_i))$. Since
  $\Update(s_1,\vec a) \simBT \FUpdate(s'_1,\vec a')$, we~get
  $\Update(s_1,(a_{-i}, b_i)) = \FUpdate(s'_1,(a'_{-i}, b'_i))$.
  
  Now, if $\datt'_j = \MaxTime +1$ for all $j \neq i$, then
  also $\datt_j \geq \MaxTime +1$ for all $j \neq i$. It~follows that,
  writing
  $(\config_3,\dat_3) = \Update(s_1,(a_{-i},b_i))$ and
  $(\config'_3,\dat'_3) = \FUpdate(s'_1,(a'_{-i},b'_i))$, we~have $\dat_3 \geq
  \MaxTime + 1$ and $\dat'_3 = \MaxTime +1$.
  Thus, $(\config_3,\dat_3) \simBT (\config'_3,\dat'_3)$.
\end{enumerate}

\item \textbf{if $\dat_1 \geq \MaxTime$:} we have that $\dat'_1 \geq
  \MaxTime$ and $\dat_2 > \MaxTime$. It~follows that the only possible
  move from $(\config'_1,\dat'_1)$ is $\vec a'$ with $a'_i = (\bot,\MaxTime+1)$
  for all $1\leq i\leq \nbp$. If $s'_2 = (\config'_2,\dat'_2) =
  \FUpdate(s'_1,\vec a')$, then $\dat'_2 = \MaxTime +1$, and thus $s_2 \simBT
  s'_2$.

  That proves the first assertion. We now prove the second
  assertion. Take $0\leq i\leq\nbp$, and $b'_i \in \FAllowed_i(s'_1)$.
  We~let $b_i = (\vertex_i,\MaxTime +1)$ for some $\vertex_i \in \Vertex$
  such that $(\vertex_i,\MaxTime +1) \in \Allowed_i(s_1)$.
    Notice that by Remark~\ref{rk-winningtraj}, 
    there always  exists $\vertex_i \in \Vertex$
  such that $\MaxTime +1
  \models \guard(\config_1(i),\vertex_i)$.
  Let $(\config_3,\dat_3) =
  \Update(s_1,(a_{-i},b_i))$ and $(\config'_3,\dat'_3) = \FUpdate(s'_1,
  (a'_{-i},b'_i))$. Since $\dat_1$ and $\dat'_1$ are larger than or equal
  to~$\MaxTime$, then $\dat_3 \geq \MaxTime
  +1$ and $\dat'_3 = \MaxTime +1$. It follows that $(\config_3,\dat_3)
  \simBT (\config'_3,\dat'_3)$.
\end{enumerate}
\end{proof}

We~now prove the converse simulation property:  
\begin{proposition}[Action vector simulation from $\FCGame$ to $\CGame$]
\label{prop:simFinToInf}
For all $s'_1 \in \FStates$, $\vec a' \in \FMov(s_1')$, and $s_1 \in \States$ such that $s_1 \simBT s'_1$, there exists $\vec a \in \Mov(s_1)$ such that 
\begin{enumerate}
\item $\Update(s_1,\vec a) \simBT \FUpdate(s'_1,\vec a')$.
\item for all $1\leq i\leq \nbp$, for all $b_i \in \Mov_i(s_1)$, there exists $b'_i \in \FMov_i(s'_1)$ such that
  \[
  \Update(s_1,(a_{-i},b_i)) \simBT \FUpdate(s'_1,(a'_{-i},b'_i)).
  \]
\end{enumerate}
\end{proposition}

\begin{proof}
  Let $s'_1 = (\config'_1,\dat'_1) \in \FStates$,
  let $\vec a'= (\vertex'_i,\datt'_i)_{1\leq i\leq\nbp} \in \FValid(s'_1)$
    (possibly with $\vertex'_i=\bot$),
  and let $s_1 = (\config_1,\dat_1) \in \States$ such that $s_1 \simBT s'_1$.
  Let $s'_2 = (\config'_2,\dat'_2) = \FUpdate(s'_1,\vec a')$.

  We have to define $\vec a = (\vertex_i,\datt_i)_{1\leq i\leq\nbp} \in
  \Valid(s_1)$ such that, letting $s_2=\Update(s_1,\vec a)$,
  then $s_2 \simBT s'_2$. This depends on~$\dat'_1$:
 
 \begin{enumerate}
 
 \item \textbf{if $\dat'_1 < \MaxTime$:} we have $s_1 = s'_1$.
   We let  $\vec a = \vec a'$.
   Notice that since $\dat'_1 < \MaxTime$, for all $1\leq i\leq \nbp$,
   we~have $\vertex'_i\not=\bot$.
   Moreover, because $\vec a' \in \FValid(s'_1)$ and $s_1= s'_1$,
   we~have $\vec a\in \Valid(s_1)$.
   In~this~way, letting $s_2 = \Update(s_1,\vec a)$, we~have
   $s_2 = s'_2$,  and in particular $s_2 \simBT s'_2$,
   which proves the first assertion. It~remains to prove the second assertion. 
 
   Pick $1\leq i\leq \nbp$ and  $b_i=(\vertex_i,\datT_i) \in \Allowed_i(s_1)$.
   We define $b'_i = (\vertex_i,\min \{ \datT_i, \MaxTime +1\})$.
   Once again, as $\dat_1=\dat'_1 < \MaxTime$ and $s_1= s'_1$,
   we~have $b'_i \in \FValid(s'_1)$.
   Let $\ora{ab}=(a_{-i},b_i)$ and
   $\ora{ab}'=(a'_{-i},b'_i)$.
   \begin{enumerate}
   \item If $\datT_i \leq \MaxTime$, then
     $\min \{ \datT_i, \MaxTime +1 \} = \datT_i$. Thus
     $\min \{ \datt'_j \mid j \neq i \} \cup \{ \datT_i \} =
     \min \{ \datt_j \mid j \neq i \}  \cup \{ \datT_i \}$,
     since for all $j \neq i$, $\datt'_j = \datt_j$.
     It~follows that $\Select(\ora{ab}') = \Select(\ora{ab})$,
     and for all $j \in \Select(\ora{ab}')$, the next vertices chosen in
     $\ora{ab}'_j$ and in $\ora{ab}_j$ are the same.
     It~follows that $\Update(s_1,\ora{ab}) = \FUpdate(s'_1,\ora{ab}')$. 

   \item If $\datT_i \geq \MaxTime +1$: if there exists $j \neq i$ such that $\datt'_j \leq \MaxTime$, then $\datt_j  = \datt'_j \leq \MaxTime$ and
     $i \not \in \Select(\ora{ab})$ nor $i \not \in \Select(\ora{ab}')$.
     It~follows that $\Update(s'_1,\vec a') = \Update(s'_1,\ora{ab}')$
     and $\FUpdate(s_1,\vec a) = \FUpdate(s_1,\ora{ab})$.
     On~the other hand, if $t'_j \geq \MaxTime +1$ for all
     $j \neq i$, then letting $(\config_3,\dat_3)= \Update(s_1,\ora{ab})$
     and $(\config'_3,\dat'_3) = \FUpdate(s'_1,\ora{ab}')$,
     we~have $(\config_3,\dat_3) \simBT (\config'_3,\dat'_3)$, 
     since $\dat_3 \geq \MaxTime +1$ and $\dat'_3 \geq \MaxTime +1$.
 \end{enumerate}

 \item \textbf{if $\dat'_1 \geq K$:}
   the only allowed move $\vec a'$ is such that $a'_i = (\bot,\MaxTime + 1)$ for all $1\leq i\leq \nbp$.
   Moreover, $\dat_1 \geq \MaxTime$.
   We then choose $\vec a$ such that $a_i = (\vertex_i,\dat'_1 + 1)$
   for some $\vertex_i \in \Vertex$ such that $(\vertex_i,\dat_1+1) \in
   \Allowed_i(s_1)$.
  Letting $s_2 = (\config_2,\dat_2) = \Update(s_1,\vec a)$,
  we~have $\dat_2 = \dat_1 + 1 > \MaxTime$. Since $\dat'_2 = \MaxTime +1$, we
  get that $s_2 \simBT s'_2$, which proves the first assertion.

  We now prove the second assertion. Take $1\leq i\leq\nbp$, and
  $b_i=(\vertex_i,\datT_i) \in \Allowed_i(s_1)$.
  Since $\dat'_1 \geq \MaxTime$, we~have $\datT_i \geq \MaxTime +1$.
  We~fix $b'_i = (\bot,\MaxTime +1)$ (there are no other possible choices).
  Letting $(\config_3,\dat_3)= \Update(s_1,(a_{-i},b_i))$
  and
  $(\config'_3,\dat'_3)= \FUpdate(s'_1,(a'_{-i},b'_i))$,
  we~have $(\config_3,\dat_3)\simBT (\config'_3,\dat'_3)$
    since $\dat_3 \geq \MaxTime +1$ and $\dat'_3 \geq \MaxTime +1$.\qed
 \end{enumerate}
 \let\qed\relax
\end{proof}

We~now establish a generic proposition showing a
correspondence between Nash equilibria in two game structures:

\propexistenceNESimu*

\begin{proof}
  Let $\sigma$ be a Nash equilibrium from~$s_0$ in~$\calG$,
  and $\rho = \Outcome{\sigma}_{s_0}$.
  By~condition~\eqref{item-cond2}, for all $1\leq i\leq \nbp$,
  we~have $\cost_i(\rho) \leq \lambda$ for some~$\lambda\in\bbN$.
  By~\eqref{item:simInit} and~\eqref{item:sim1}, there exists $\rho' \in \Plays_{\calG'}(s'_0)$ such that $\rho \genSim \rho'$. By~\eqref{item:simCost}, for all $1\leq i\leq\nbp$, $\cost_i(\rho) = \cost'_i(\rho')$.
We will construct a strategy profile $\sigma'$ in $(\calG'$ that is a Nash equilibrium from~$s'_0$ and such that $\Outcome{\sigma'}_{s'_0} = \rho'$.

For any history $h=((s_k,\vec a_k,s'_k)_{0\leq k < \ell}$ in $\Hist_{\calG}(s_0)$,
for any $1\leq i\leq\nbp$, and for any strategy profile~$\sigma$,
we~define $\Devi(h,i,\sigma)$ to be \True
if, and only~if,
for all $0\leq k<\ell$ and all $j \neq i$,
$a_{k,j} = \sigma_j(\rho_{<k})$.
By~convention, we~let $\Devi(s_0,i,\sigma)$ be \True.
We define $\Devi$ in the same way on histories of $\calG'$.
For all $h \in \Hist_{\calG}$ and all strategy profiles $\sigma$, we~let
\[
\setDevi(h,\sigma) = \{ 1\leq i\leq \nbp  \mid
\Devi(h,i,\sigma)= \True \}.
\]
Notice that $\setDevi(h,\sigma)$ can be either $\Players$, in case no
players deviated, or a singleton~$\{i\}$, in case only Player~$i$
deviated, or empty, in case at least two players deviated.

We will construct $\sigma'$ step-by-step as follows:
\emph{(i)}~we~define $\sigma'$ such that its outcome is $\rho'$ and
\emph{(ii)} we extend $\sigma'$ on histories, by induction on their length.
During the construction of $\sigma'$, we define a partial
fonction $\repr\colon \Hist_{\calG'}(s'_0) \longrightarrow
\Hist_{\calG}(s_0)$, which associates, with
each history~$h'\in\Hist_{\calG'}(s'_0)$ such that
$\setDevi(h'c',\tau)$ is non-empty, a~representative history in~$\calG$.
At~the end of the procedure we want that
$\repr$ satisfies the following properties:
for all $h'$ of length~$k$
in $\Hist_{\calG'}(s'_0)$ such that $\setDevi(h',\sigma')$ is non-empty:
\begin{enumerate}\def\theenumi{P\arabic{enumi}}
\item \label{item:repr1}  for all $\ell\leq k$,
  $\repr(h'_{<\ell}) = \repr(h')_{<\ell}$; 
\item \label{item:repr2} $\setDevi(h',\sigma')=\setDevi(\repr(h'),\sigma)$ and
  $\repr(h') \genSim h'$. 
\end{enumerate}

We first let $\sigma'_i(\rho'_{<k})=a'_{k,i}$
and $\repr(\rho'_{<k})= \rho_{<k}$
for all~$1\leq i\leq\nbp$ and all $k\in\bbN$,
so that $\Outcome{\sigma'}_{s'_0} = \rho'$. Thus
by~construction, for all $k \in \N$, $\setDevi(\rho'_{<k},\sigma') =
\Players$ and $\rho'_{<k}$ satisfies Properties~\eqref{item:repr1}
and~\eqref{item:repr2}.

We now extend the definition of $\sigma'$ step-by-step. At~step~$k$:
\begin{enumerate}
\item we define $\sigma'$ on all histories $h' \in
  \Hist_{\calG'}(s'_0)$ of length~$k$ that are not prefixes
  of~$\rho'$, and $\repr$ on the resulting outcomes, of length~$k+1$;
\item we prove that Properties~\eqref{item:repr1}
  and~\eqref{item:repr2} are satisfied for all histories $h'$ of
  length $k+1$ for which $\setDevi(h',\sigma')$ is non-empty.
\end{enumerate}

\paragraph{\textbf{At step 0:}}
The only history in $\Hist_{\calG'}(s'_0)$ of length~$0$ is $s'_0$, which is
a prefix of~$\rho'$.
By~hypothesis, %
we have that $s_0 \genSim s'_0$.
By~construction, $\repr(s'_0) = s_0$.

We define $\repr$ on all histories $h'=(u_0,\vec a',u'_0)$ of
length~$1$ that are not prefixes of~$\rho'$ and for which
$\setDevi(h', \sigma')$ is non-empty (hence $\setDevi(h', \sigma')$ is
a singleton~$\{i\}$ for some~$1\leq i\leq\nbp$, since there must have
been a deviation for~$h'$ not to be a prefix of~$\rho'$).

By~Condition~\eqref{item:sim2}, there exists $b_i \in \Allowed_i(s_0)$
such that $s_1 = \Update(s_0, (\sigma_{-i}(s_0),b_i))$ is such that
$v_1 \genSim v'_1$. We~let $\repr(h') = (s_0, (\sigma_{-i}(s_0),b_i), s_1)$.

Property~\eqref{item:repr1} clearly holds, and
since $\setDevi(h',\sigma') = \{i\} = \setDevi(\repr(h'),\sigma)$
and $\repr(h') \genSim h'$, Property~\eqref{item:repr2} is also satisfied.

\paragraph{\textbf{At step $k$:}}
Let us assume that step $k-1$ has been completed, and that
Properties~\eqref{item:repr1} and~\eqref{item:repr2} hold for
histories of length~$k$.

We~first define $\sigma'$ on histories of length~$k$ that are not
prefixes of~$\rho'$. For such a history~$h'=(u'_j,\vec
a'_j,v'_j)_{0\leq j<k}$, $\setDevi(h',\sigma')$ is either empty (if at
least two players deviated from~$\sigma'$), or it is a
singleton~$\{i\}$ (in~case only Player~$i$ deviated).

In the first case ($\setDevi(h',\sigma')=\emptyset$), we~let
$\sigma'(h')=\vec b'$ for some $\vec b'\in\Valid'(\last(h'))$.
In~the second case ($\setDevi(h',\sigma')=\{i\}$), we~let $h =
(u_j,\vec a_j,v_j)_{0\leq j<k}=\repr(h')$. 
By~induction hypothesis, we~have $h\genSim h'$ and $\setDevi(h,\sigma) =\{i\}$.
Take $\vec b=\sigma(h)$: by~\eqref{item:sim1}, there exists $\vec b'\in
\Valid(u'_k)$ such that, writing $v_{k}=\Update(u_k,\vec b)$ and
$v'_{k}=\Update(u'_k,\vec b')$, $v_k \genSim v'_k$
(remember that we~write $u_k$ for~$v_{k-1}$). We~then let
$\sigma'(h')=\vec b'$, and
$\repr(h'\cdot (u'_k,\vec b',v'_k)) = h\cdot (u_k,\vec b,v_k)$.

We~now prove that Properties~\eqref{item:repr1} and~\eqref{item:repr2}
are satisfied for all histories~$h'$ of length~$k+1$ for which
$\setDevi(h',\sigma')$ is non-empty.
Let $h' = (u'_j,\vec a'_j,v'_j)_{0\leq j<k+1}$
be a history of length~$k+1$ in $\Hist_{\calG'}(s'_0)$.

If $\setDevi(h',\sigma') = \Players$, then $h'$ is a prefix of
$\rho'$, and we already proved that Properties~\eqref{item:repr1}
and~\eqref{item:repr2} are satisfied.

We now focus on the case where $\setDevi(h',\sigma') = \{ i \}$ for
some $1\leq i\leq\nbp$:
\begin{itemize}

\item if $\repr(h')$ was already defined, then the last move was
  played by all players according to~$\sigma'$.
  Thus $\setDevi(h'_{<k},\sigma') = \{i\}$ and for all $j \in
  \Players$, $a'_{k} = \sigma'(h'_{<k})$. 
  By~definition of $\repr(h')$, Properties~\eqref{item:repr1}
  and~\eqref{item:repr2} are satisfied.

\item if $\repr(h')$ is not yet defined,
  it~means that $a'_{k}(i) \neq \sigma'_i(h'_{<k})$
  and $a'_{k}(j) = \sigma'_j(h'_{<k})$ for all $j \neq i$.
  By~induction hypothesis,
  $\setDevi(h'_{<k},\sigma') =
  \setDevi(\repr(h'_{<k}),\sigma)$
  and
  $\repr(h'_{<k}) \genSim h'_{<k}$.

  By~\eqref{item:sim2}, there exists $b_i \in \Allowed_i(u_k)$
  such that $v_{k} = \Update(u_k, (\sigma_{-i}(h),b_i) \genSim  v'_{k}$.

  We fix $\repr(h') = \repr(h'_{<k}) \cdot (u_k, (\sigma_{-i}(h),b_i), v_{k})$.
  By the two properties above, 
  Properties~\eqref{item:repr1} and~\eqref{item:repr2} are satisfied for $h'$.
\end{itemize}

This concludes the construction of $\sigma'$ such that
$\Outcome{\sigma'}_{s'_0} = \rho'$.  It remains to prove that $\sigma'$ is
a Nash equilibrium in $\calG'$ from~$s'_0$.
Towards a contradiction, we assume that there exists $1\leq i\leq\nbp$ and a strategy~$\tau'_i$ of Player~$i$ in $\calG'$ from~$s'_0$ such that 
$ \cost_i(\Outcome{\sigma'_{-i},\tau'_i}_{s'_0}) <
  \cost_i(\Outcome{\sigma'}_{s'_0})$.
In particular, %
$\cost_i(\Outcome{\sigma'_{-i},\tau'_i}_{s'_0}) < \lambda$.

Let $\overline{\rho'}= \Outcome{\sigma'_{-i},\tau'_i}_{s'_0}$.
We~will build $\overline{\rho} \in \Plays_{\calG}(s_0)$
such that $\overline{\rho} \genSim \overline{\rho'}$ and
for all $k \in \N$, $\setDevi(\overline{\rho}_{<k},\sigma) = \setDevi(\overline{\rho'}_{<k}, \sigma')$.  
This entails that there exists a strategy~$\tau_i$ of Player~$i$
in $\calG$ from~$s_0$ such that
$\Outcome{\sigma_{-i},\tau_i}_{v_0} = \overline{\rho}$. It follows that 
\[
\cost_i(\overline{\rho}) = \cost_i(\overline{\rho'})  <
\cost_i(\rho') = \cost_i(\rho).
\]
Thus $\tau_i$ is a profitable deviation for Player~$i$ w.r.t.~$\sigma$.
which contradicts the fact that $\sigma$ was chosen to be a Nash equilibrium.

\smallskip

It remains to show to construct $\overline{\rho}$ from $\overline{\rho'}$.
We~proceed by induction.
\begin{itemize}
\item we~let $\overline{\rho}_0 = s_0$; by hypothesis $s_0 \genSim s'_0$.
  Moreover, $\setDevi(\overline{\rho}_{<0},\sigma) = \Players$ and
  $\setDevi(\overline{\rho'}_{<0},\sigma') = \Players$;
\item We assume that for all $0 \leq \ell \leq k$, we~have
  $\overline{\rho}_{< \ell} \genSim \overline{\rho'}_{< \ell}$,
  $\setDevi(\overline{\rho}_{<\ell},\sigma) =
  \setDevi(\overline{\rho'}_{<\ell}, \tau)$ and
  $\overline{\rho}_{< \ell} = \repr(\overline{\rho'}_{< \ell})$.

  We show how to define the last step of $\overline{\rho}_{<k+1}$.
  Notice that for
  all prefix $h' < \overline{\rho'}$, $\repr(h')$ is well-defined,
  since $\setDevi(h',\sigma')$ is either $\Players$ or~$\{i\}$.
  Let $h= (u_j,\vec a_j,v_j)_{0\leq j<k+1}= \repr(\overline{\rho'}_{< k+1})$. 
  By~construction, $\setDevi(h,\sigma) =
  \setDevi(\overline{\rho'}_{< k+1}, \sigma')$ and
  $h \genSim \overline{\rho'}_{< k+1}$.
  Thus we choose
  $\overline{\rho}_{<k+1} = \overline{\rho}_{<k} . (u_{k},\vec a_k, v_{k})$.
\end{itemize}

In this way, we obtain that for all $k \in \N$,
$\overline{\rho'}_{< k} \genSim \overline{\rho}_{< k}$ and
$\setDevi(\overline{\rho}_{< k}, \sigma) =
\setDevi(\overline{\rho'}_{< k}, \sigma')$.
\end{proof}

\begin{proof}[Proof of Theorem~\ref{thm:equivNE}]
  Theorem~\ref{thm:equivNE} can now be proven by applying Proposition~\ref{prop:existenceNESimu} twice.
  We~first apply Proposition~\ref{prop:existenceNESimu} to $\CGame$
  and $\FCGame$, 
  with $\genSim = \simBT$.

  We prove that $\CGame$, $\FCGame$ and $\simBT$
  satisfy the four hypotheses of Proposition~\ref{prop:existenceNESimu},
  which will allow us to conclude
  that any Nash equilibrium in~$\CGame$ from~$s_0=\start$ has a
  corresponding Nash equilibrium from~$s'_0=\start$ in~$\FCGame$,
  with the same cost profile:
  \begin{enumerate}
  \item by definition of~$\simBT$, we have $s_0 \simBT s'_0$;
    
  \item we choose $\lambda = \MaxTime$, and
    the assertion holds by Lemma~\ref{lemma:boundedLength}.

  \item let $\rho = (u_k, \vec a_k, v_k)_{k \in\bbN} \in \PlaysG(s_0)$
    and $\rho' = (u'_k,\vec a'_k, v'_k)_{k\in\bbN} \in \PlaysFG(s'_0)$
    such that $\rho \simBT \rho'$. Let $1\leq i\leq\nbp$.
    \begin{itemize}
    \item If $\cost_i(\rho) \leq \MaxTime$, by~Lemma~\ref{lemma:boundedLength}
      there exists $k^*_i \leq \MaxTime$ such that
      $u_{k^*_i} = (\tgt_i, \dat_i)$ with $\dat_i\leq\MaxTime$.
      It~follows that for all $\ell \leq k_i^*$, $\rho_{<\ell} = \rho'_{<\ell}$,
      and $\cost_i(\rho) = \cost_i(\rho')$.
    \item If $\cost_i(\rho') \leq \MaxTime$, applying
      Lemma~\ref{lemma:boundedLength} to plays in~$\FCGame$,
      the~same reasoning applies.
    \end{itemize} 
  \item by  Proposition~\ref{prop:simInfToFin}.
\end{enumerate}

  Conversely, 
  we~apply Proposition~\ref{prop:existenceNESimu} to $\FCGame$ and $\CGame$,
  with~$\genSim=\simBT^{-1}$.
  We have to prove that
  they satisfy the four conditions of Proposition~\ref{prop:existenceNESimu}:
\begin{enumerate}
\item we have $s_0 \simBT s'_0$;
\item we choose $\lambda = \MaxTime$ and the assertion holds by Lemma~\ref{lemma:boundedLength}.
\item Same proof as the other implication.
\item By Proposition~\ref{prop:simFinToInf}.
\end{enumerate}
This concludes the proof of Theorem~\ref{thm:equivNE}.
\end{proof}

\subsection{Proofs of Section~\ref{sec-ExistNE}}

\begin{restatable}{lemma}{lemmacostpath}\label{lemma:costTimedPath}
    For all $\sigma_i$ and~$\tau_i$ in~$\Sigma_i(s_0)$,
    for all $\sigma_{-i} \in \Sigma_{-i}(s_0)$ such that for all $j \in \Players\backslash \{i \}$, $\sigma_j$~is a blind strategy, we~have: 
   if $\traj_i(\Outcome{\sigma_{-i},\sigma_i}_{s_0}) =
    \traj_i(\Outcome{\sigma_{-i},\tau_i}_{s_0})$, then
   $\cost_i(\Outcome{\sigma_{-i},\sigma_i}_{s_0}) =
   \cost_i(\Outcome{\sigma_{-i},\tau_i}_{s_0})$.
\end{restatable}

\begin{proof}
  The trajectories are the same for all players in both outcomes,
  hence the sequences of timed configurations visited along both
  outcomes are the same. Now, the cost of a transition only depends on
  the timed configurations it comes from and goes~to, and not on the
  exact action vector that is being played. The result follows.
\end{proof}

\propENrestToEN*
\begin{proof}
  Let us assume that there exists a $\Sigma^{\WB}$-Nash equilibrium~$\sigma$
  in $\FCGame$ from~$\start$. Notice that for all $1\leq i\leq\nbp$,
  strategy~$\sigma_i$ is a winning blind strategy for Player~$i$ from~$\start$.
  We prove that $\sigma$ is also a Nash equilibrium from~$\start$ in~$\FCGame$. 
	
  Ad~absurdum, we assume that there exist $1\leq i\leq\nbp$
  and a strategy~$\tau_i$  of Player~$i$ in $\FCGame$ from~$\start$ such that 
  \begin{equation}
    \Fcost_i(\Outcome{\sigma_{-i},\tau_i}_{\start}) <
    \Fcost_i(\Outcome{\sigma}_{\start}). \label{eq:existenceNE}
  \end{equation}

  Consider the trajectory $p =
  \traj_i(\Outcome{\sigma_{-i},\tau_i}_{\start})$.
  Writing~$p =
      (\vertex_k,\dat_k)_{k \in\bbN}$, since
      $\cost_i(\Outcome{\sigma_{-i},\tau_i}_{\start})$ is finite, there
  exists $k^* \in \bbN$ such that $\vertex_{k^*} = \tgt_i$.
  Thus there exists
  a winning blind strategy~$\sigma'_i \in \Sigma^{\WB}_i$
  corresponding to this trajectory.
		
  By hypothesis $\sigma_j$ is a blind strategy for all $j \in \Players \backslash \{i\}$. By~Lemma~\ref{lemma:costTimedPath}, it~follows:
  \[
  \Fcost_i(\Outcome{\sigma_{-i}, \sigma'_{i}}_{\start}) = \Fcost_i(\Outcome{\sigma_{-i},\tau_i}_{\start}).
  \]
		
  Thus $\sigma'_i \in \Sigma^{\WB}_i$ is a profitable blind deviation
  for Player~$i$ w.r.t.~$\sigma$,
  contradicting the fact that $\sigma$ is a $\Sigma^{\WB}$-Nash equilibrium.
\end{proof}

\proppotentialGame*
\begin{proof}

  Let $R=\{(\vertex,\dat) \mid \vertex\in\Vertex, 0\leq\dat\leq
  \MaxTime\}$. For each strategy profile~$\sigma\in\Sigma^{\WB}$ and
  each pair $(\vertex,\dat)\in R$, writing
  $\InitOutcome{\sigma}{\start}= ((\config_k,\dat_k),\vec
  a_k,(\config'_k,\dat'_k))_{k\in\bbN}$ and $k^*$ for the largest
  index for which $\dat_k\leq\dat$, we~define $\load_\sigma(\vertex,\dat) =
  \#\{1\leq i\leq\nbp \mid \config_{k^*}(i)=\vertex\}$.  In~other
  terms, $\load_\sigma(\vertex,\dat)$ is the number of players standing
  in vertex~$\vertex$ at time~$\dat$ along the outcome of~$\sigma$
  from~$\start$.

  For any strategy profile $\sigma$ in $\Sigma^{\WB}$ and any subset~$J\subseteq R$, we define
  \[
  \Psi_J(\sigma) = \sum_{(\vertex,\dat)\in J} \sum_{k = 1}^{\load_{\sigma}(v,\dat)} \weight(\vertex)(k).
  \]

  We~let $\Psi=\Psi_R$, and prove that it is a potential function: 
  assume that some player~$i$ deviates from their strategy $\sigma_i$ in~$\sigma$ and follows some other strategy~$\sigma'_i\in \Sigma^{\WB}_i$ instead;
  write $\sigma'=(\sigma_{-i}, \sigma'_i)$.
  We~prove that:
  \[
  \Psi(\sigma') - \Psi(\sigma) =
  \Fcost_i(\InitOutcome{\sigma'}{\start}) -
  \Fcost_i(\InitOutcome{\sigma}{\start}).
  \]

  Write~$P$ for the set of plays in~$\FCGame$ along which all players
  visit their target vertices. For such a
  play~$\rho=((\config_k,\dat_k),\vec
  a_k,(\config'_k,\dat'_k))_{k\in\bbN}$ in~$P$, and for each~$1\leq
  i\leq\nbp$, we~define the set of positions occupied by Player~$i$
  along~$\rho$ (before reaching their target vertex) as follows:
  we~first let $k_i$ be the least index~$k$ such that $\config_k(i)=\tgt_i$;
  then
  \[
  \Pos_i(\rho) = \{(\vertex,\dat) \mid \exists 0\leq j<k_i.\
  \dat_{j}\leq \dat <\dat_{j+1} \text{ and } \vertex = \vertex_j\}.
  \]

  Now, let $V=\Pos_i(\InitOutcome{\sigma}{\start})$ and
  $V'=\Pos_i(\InitOutcome{\sigma'}{\start})$. Notice
  that since $\sigma$ and $\sigma'$ are in $\Sigma^{\WB}$,
  both~$V$ and~$V'$ are well-defined.
  Then let
  \begin{xalignat*}4
    B &= V \cap V' &
    O &= V\setminus B &
    N &= V'\setminus B &
    X &= R\setminus (B\cup O \cup N).
  \end{xalignat*}
  Notice that $X$, $B$, $O$ and~$N$ are disjoint.
  Moreover, since all players but Player~$i$ follow the same
  blind strategy in~$\sigma$ and in~$\sigma'$, we~have
  \begin{xalignat*}2
    \load_{\sigma}(\vertex,\dat) &= \load'_{\sigma}(\vertex,\dat)
      && \text{for all $(\vertex,\dat) \in B \cup X$} \\
    \load_{\sigma}(\vertex,\dat) &= \load'_{\sigma}(\vertex,\dat) + 1
      && \text{for all $(\vertex,\dat) \in O$} \\
    \load_{\sigma}(\vertex,\dat) &= \load'_{\sigma}(\vertex,\dat) -1
      && \text{for all $(\vertex,\dat) \in N$.} 
  \end{xalignat*}

  We can then write:
  \begin{xalignat*}1
    \Fcost_i(\InitOutcome{\sigma'}{\start})-
      \Fcost_i(\InitOutcome{\sigma}{\start}) &=
    \sum_{(\vertex,\dat)\in V'} \weight(\vertex)(\load_{\sigma'}(\vertex,\dat)) -
      \sum_{(\vertex,\dat)\in V} \weight(\vertex)(\load_{\sigma}(\vertex,\dat)) \\
   &= \sum_{(\vertex,\dat)\in N} \weight(\vertex)(\load_{\sigma'}(\vertex,\dat)) -
      \sum_{(\vertex,\dat)\in O} \weight(\vertex)(\load_{\sigma}(\vertex,\dat)).
  \end{xalignat*}
  Indeed, the terms $\sum_{(\vertex,\dat)\in B}
  \weight(\vertex)(\load_{\sigma}(\vertex,\dat))$
  and
  $\sum_{(\vertex,\dat)\in B}
  \weight(\vertex)(\load_{\sigma'}(\vertex,\dat))$ cancel out
  by the first equality above.

  Similarly, 
  \begin{xalignat*}1
    \Psi(\sigma')-\Psi(\sigma) &=
      \Psi_X(\sigma')+\Psi_B(\sigma')+\Psi_O(\sigma')+\Psi_N(\sigma') - 
      (\Psi_X(\sigma)+\Psi_B(\sigma)+\Psi_O(\sigma)+\Psi_N(\sigma)) \\
      &=  \sum_{(\vertex,\dat)\in N} \weight(\vertex)(\load_{\sigma'}(\vertex,\dat)) -
      \sum_{(\vertex,\dat)\in O} \weight(\vertex)(\load_{\sigma}(\vertex,\dat)).
  \end{xalignat*}
  The equality proving that $\Psi$ is a potential function follows.
\end{proof}

We can now prove Theorem~\ref{thm:existenceNE}:
\begin{proof}%
  Let $\timedNetGame$ be a timed network game and let $\FCGame$ be its
  finite associated concurrent game.  Take a winning blind strategy
  profile~$\sigma\in\Sigma^{\WB}$ (which must exist thanks to our
  hypotheses). Then $\Psi(\sigma)$ is finite. Then, as long as is
  possible, replace the winning blind strategy of some player~$i$ with
  one that achieves a (strictly) better cost for that player. From
  Proposition~\ref{prop:potentialGame}, the~value of~$\Psi$ (strictly)
  decreases at each step. Since~$\Psi$ takes non-negative integer
  values, this process must terminate. Upon convergence, the strategy
  profile we~obtain is a winning blind strategy profile in which no
  player can improve their cost by a unilateral deviation, hence it
  is a $\Sigma^{\WB}$-Nash equilibrium. By~Proposition~\ref{prop:ENrestToEN}, it~is
  a Nash equilibrium in~$\FCGame$, and by Theorem~\ref{thm:equivNE},
  it is also a Nash equilibrium in~$\CGame$.
\end{proof}

\subsection{Proofs of Section~\ref{sec-computeNE}}

\thmcritNE*
\begin{proof}
  We begin with the first direction, showing that the
  outcome of a Nash equilibrium satisfies Property~\eqref{eq-NE}.
  
  Let $\sigma$  be a Nash equilibrium in $\FCGame$,
  and $\rho = \InitOutcome{\sigma}{\start}$.
  Let $1\leq i\leq\nbp$ and $k \in \N$
  such that $i \not\in \Visit(\rho_{<k})$,
  and $b_i \in \FAllowed_i(s_k)$.
  Let $s' = \Update(s_k, (\vec a_{k,-i}b_i))$ be the vertex reached
  after the deviation of Player~$i$ in $s_k$.

  We~define a set~$D_i\subseteq \Sigma_i(\start)$ of all strategies for
  Player~$i$ that follow~$\sigma_i$ along the first $k-1$ steps of~$\rho$,
  and play~$b_i$ at~$s_k$. When~$k=0$, $D_i$ contains all strategies that
  play~$b_i$ from~$\start$.

  Let $h = \rho_{<k}$ be the prefix of $\rho$ until $s_k$,
  and $h'= h\cdot (s_k,(\vec a_{k,-i},b_i),s')$.
  Since $\sigma$ is a Nash equilibrium from~$\start$ in $\FCGame$,
  we have that for all $\tau_i \in D_i$:
  \[
  \Fcost_i(h\cdot \InitOutcome{\sigma_{\restriction h}}{s_k}) =
  \Fcost_i(\InitOutcome{\sigma}{\start}) \leq
  \Fcost_i(\InitOutcome{\sigma_{-i}}{\start},\tau_i ) =
\Fcost_i(h\cdot\InitOutcome{\sigma_{-i \restriction h},\tau_{i \restriction h}}{\start})
  \]
  where $\sigma_{\restriction h}$ is the residual strategy of~$\sigma$ after
  history~$h$.
  Since $i \not\in \Visit(h)$, we~get $\cost_i(\InitOutcome{\sigma_{\restriction h}}{s_k}) \leq \cost_i(\InitOutcome{\sigma_{-i \restriction h},\tau_{i \restriction h}}{\start})$.
  It~follows that,  for all $\tau_i \in D_i$:
  \begin{xalignat*}1 
    \Fcost_i(\InitOutcome{\sigma_{\restriction h}}{s_k}) 
    &\leq \Fcost_i(\InitOutcome{\sigma_{-i \restriction h},\tau_{i \restriction h}}{s_k})\\
    &= \Fcost_i((s_k,(\vec a_{k,-i},b_i),s')\cdot \InitOutcome{\sigma_{-i \restriction h'},\tau_{i \restriction h'}}{s'}) \\
    &= \Fcost_i(s_k,(\vec a_{k,-i},b_i),s') + \Fcost_i(\InitOutcome{\sigma_{-i \restriction h'},\tau_{i \restriction h'}}{s'}). %
  \end{xalignat*}
Since this must hold for any strategy~$\tau_i$ in~$D_i$, we~get:
\begin{xalignat}1
  \Fcost_i(\InitOutcome{\sigma_{\restriction h}}{s_k}) 
  & \leq \Fcost_i(s_k,(\vec a_{k,-i},b_i),s') +
     \inf_{\tau_i \in D_i} \Fcost_i(\InitOutcome{\sigma_{-i \restriction h'}, \tau_{i \restriction h'}}{s'}) \notag\\
  & = \Fcost_i(s_k,(\vec a_{k,-i},b_i),s') +
     \inf_{\tau_i \in \Sigma_i} \Fcost_i(\InitOutcome{\sigma_{-i \restriction h'}, \tau_{i \restriction h'}}{s'}) \label{eq:critEN3}  \\
     & \leq \Fcost_i(s_k,(\vec a_{k,-i},b_i),s') +
     \sup_{\nu_{-i}\in \Sigma_{-i}} \inf_{\tau_i \in \Sigma_i} \Fcost_i(\InitOutcome{\tau_{-i \restriction h'},\nu_{i \restriction h'}}{s'}) \notag\\
     & = \Fcost_i(s_k,(\vec a_{k,-i},b_i),s') +\LowVal_i(s').\notag
\end{xalignat}
Equality~\eqref{eq:critEN3} holds because only the residual part of~$\tau_i$ after~$h'$ is used, and there are no constraints on that part in~$D_i$.
This concludes the proof of this implication.

\medskip

We now turn to the converse implication.  Take $\rho=(s_k,\vec a_k,s'_k)_{k\in\bbN} \in
\PlaysFG(\start)$ satisfying Property~\eqref{eq-NE}, We~build a
strategy profile~$\sigma$ from~$\start$ in $\FCGame$ that is a Nash
equilibrium and such that $\InitOutcome{\sigma}{\start} = \rho$.

The main idea is that all players follow~$\rho$, until one player
deviates. If we assume that Player~$i$ is the first player who
deviates, then after this deviation the strategy of the coalition~$-i$
of the other players will follow the punishing strategy
in the corresponding two-player zero-sum game.

For all $1\leq i\leq\nbp$, for all $j \neq i$, and for all state~$s$
in~$\FCGame$, we denote by $\sigma^s_{j,i}$ the strategy of Player~$j$
obtained from the optimal strategy of coalition $-i$ from~$s$ in the
two-player zero-sum game in which Player~$i$ aims at minimizing their
cost. In~such two-player zero-sum games with weights in~$\bbN$, such
optimal strategies exist,
and 
for any state~$s$, we~have
\begin{equation}
  \inf_{\mu_i \in \Sigma_i(s)} \Fcost_i(\InitOutcome{(\sigma^s_{j,i})_{j \neq i},\mu_i}{s}) \geq \LowVal_i(s). \label{eq:valStratOpti}
\end{equation}

First, we define $P_\rho\colon \HistFG(\start) \to (\HistFG(\start)\times \Players) \cup \{ \bot \}$
to keep track of the first deviation (w.r.t.~$\rho$) along a history~$h$.
Formally, for $h=(u_k,\vec b_k,u'_k)_{0\leq k< \ell}$,
we~let
\[
P_\rho (h)  =
\begin{cases}
  P_{\rho}(h_{<\ell-1}) & \text{ if $h_{<\ell-1}$ is not a prefix of~$\rho$} \\
  (h,i) & \text { if  $h_{<\ell-1}$ is a prefix of $\rho$ but $h$
    is not, and} \\
  & \qquad \text{ $i$ is the least index for which $b_{\ell,i}\not=a_{\ell,i}$} \\
  \bot & \text{otherwise} %
\end{cases}
\]
Hence $P_{\rho}(h)=(h',i)$ indicates that there has been a deviation w.r.t.~$\rho$ along~$h$, and that the first deviation occurred after prefix~$h'$, and that Player~$i$ is one of the players who deviated at that point. 
If there has been no deviation (hence $h$ is a prefix of~$\rho$), then $P_{\rho}(h)=\bot$.

We can now define~$\sigma$:
for all $h\in \HistFG(\start)$ of length~$\ell$
and for all $i \in \Players$, we~let
\[
\sigma_i(h) =
\begin{cases}
  a_{\ell,i} & \text{ if $h$ is a prefix of~$\rho$} \\
  \text{any allowed move} & \text{ if $P_{\rho}(h) =(h',i)$ for some~$h'$} \\
  \sigma^s_{i,j}(h\setminus h') & \text{ if $P_{\rho}(h) = (h',j)$ for some $h'$ and some $j \neq i$, } \\
  & \text{ with $s=\last(h')$ and $h\setminus h'$ is the suffix of~$h$ after~$h'$}
\end{cases}
\]

Let us prove that $\sigma$ is a Nash equilibrium in~$\FCGame$.
Let $1\leq i\leq \nbp$, let $\tau_i \in \Sigma_i(\start)$.
Let $\rho= (s_k,\vec a_k,s'_k)_{k\in\bbN}=\InitOutcome{\sigma}{\start}$ and
$\rho'= (u_k,\vec b_k,u'_k)_{k\in\bbN} = \InitOutcome{\sigma_{-i},\tau_i}{\start}$.
We~have to prove that 
\(
\Fcost_i(\rho) \leq
\Fcost_i(\rho')
\).

If both outcomes are identical, the result is trivial.
If~not, let $k \in \N$ be the largest index such that
$\rho_{< k} = \rho'_{<k}$: this means that the first deviation
of Player~$i$ occurs in state~$s_k$ along~$\rho$.

If $i \in \Visit(\rho_{< k})$, then $\Fcost_i(\rho) = \Fcost_i(\rho')$,
and we are done. Otherwise:
\bgroup
\makeatletter
\def\tagform@#1{\maketag@@@{\ignorespaces#1\unskip\@@italiccorr}}
\begin{xalignat*}1
  \Fcost_i(\rho') & =
  \Fcost_i(h'\cdot \InitOutcome{\sigma_{-i \restriction h'}, \tau_{i\restriction h'}}{\rho'_{k+1}}) \tag{where  $h' = \rho'_{<k+1}$} \\
  & = \Fcost_i(h' \cdot \InitOutcome{(\sigma^{u_{k+1}}_{j,i \restriction h'})_{j\not=i}, \tau_{i\restriction h'}}{u_{k+1}}) \tag{by definition of $\sigma$} \\
  &= \Fcost_i(\rho_{<k}) + \Fcost_i(u_k, \vec b_k,u'_{k}) + \Fcost_i(\InitOutcome{ (\sigma^{u_{k+1}}_{j,i \restriction h'})_{j \neq i},\tau_{i\restriction h'}}{u_{k+1}}) \\
  \noalign{\hfill because $i \not \in \Visit(\rho_{<k})$}
  & \geq \Fcost_i(\rho_{< k}) + \Fcost_i(u_k, \vec b_k,u'_{k}) +
      \inf_{ \nu_i \in \Sigma_i(u_{k+1})} \Fcost_i(\InitOutcome{(\sigma^{u_{k+1}}_{j,i})_{j \neq i},\nu_i}{u_{k+1}})\\
  & \geq \Fcost_i(\rho_{< k}) + \Fcost_i(u_k, \vec b_k,u'_{k}))  +
      \LowVal_i(u_{k+1}) 
      \tag{by \eqref{eq:valStratOpti}} \\
  & \geq \Fcost_i(\rho_{<k}) + \Fcost_i(\rho_{\geq k}) \tag{by~\eqref{eq-NE}} \\
 &= \Fcost_i(\InitOutcome{\sigma}{\start}) \tag{\qed}
\end{xalignat*}
\egroup
\let\qed\relax
\end{proof}

\begin{remark}
Our definition of plays and histories include the action vectors at each
step, which implies that strategies observe the actions of all players
and can base their decisions on those informations. In~particular,
they can easily detect changes in strategies. This is used in the
definition of $P_\rho$, and thus of~$\sigma_i$, in~the proof of
Theorem~\ref{thm:critNE}.

We could easily adapt our results to handle the case where strategies
are not allowed to depend on action vectors. Indeed, in our setting of
network congestion games, the~relevant single-player deviations can be
detected by observing only the sequence of configurations. With
\emph{relevant}, we~mean those deviations that modify the outcome
(other deviations would never be profitable). 

Indeed, assume that a single player deviates from their strategy from
some configuration~$(\config,\dat)$, so that the play goes
to~$(\config_2,\dat_2)$ instead of~$(\config_1,\dat_1)$.  In~case
$\dat_2<\dat_1$, then $\dat_2$ must be the date proposed by the
deviating player, and they will be the only player applying their
action; they are thus easily identified.  In~case $\dat_2>\dat_1$,
it~must be the case that $\dat_1$ was proposed by the deviating player
alone, and they were the only player applying their action
between~$\config$ and~$\config_1$.  Finally, in~case $\dat_2=\dat_1$,
by~hypothesis $\config_1\not=\config_2$, so if only one player, say
Player~$i$, deviated, then $i$~will be the only index for which
$\config_1(i)\not=\config_2(i)$, which again permits to identify the
deviating player.

In~the end, in case strategies are not allowed to depend on move
vectors, the proof of Theorem~\ref{thm:critNE} can be adapted, and our
result still holds. 
\end{remark}

\begin{restatable}{proposition}{propcomputeVal}
  For all $1\leq i\leq\nbp$ and for all $s \in \FStates$,
  the value $\LowVal_i(s)$ can be computed in exponential time.
\end{restatable}
\begin{proof}
  From the finite concurrent game $\FCGame$, we build a two-player
  zero-sum concurrent game~$\calZ$ in which the two players are the
  player~$i$, which has the same possible actions as in~$\FCGame$, and
  the coalition of other players~$-i$, which chooses the actions of
  all the other players. The states of~$\calZ$ are $\Vertex \times [0,\nbp]^{\Vertex}\times [0, \MaxTime +1]$: such a state
  keeps track of  the current position of Player~$i$, an abstract representation of the current position of all
  players of the coalition~$-i$ (\emph{i.e.},~the~number of players of
  the coalition in each vertex) and the current date. The~objective of Player~$i$ is to
  reach~$\tgt_i$ (starting from~$\src_i$) while minimizing their cost,
  and the goal of coalition player~$-i$ is opposite, \emph{i.e.}, to
  maximize the cost of Player~$i$.

  Computing~$\LowVal_i$ in such a game can be performed in 
  time polynomial in the size of the game~\cite{LMO06},
  hence in time exponential
  in~$|\Vertex|$ and pseudo-polynomial in~$\MaxTime$.
\end{proof}

\section{Proofs of Section~\ref{section:socOptiPriceAnarch}}
\label{ann-section4}

\propconstSWEXPSPACEasym*

\begin{proof}
Let $\timedNetGame$ be a timed network game and let $x \in \N$ be a threshold. We want to decide if there exists a play~$\rho$ in~$\timedNetGame$ beginning in $\start$ such that $\SW(\rho) \leq x$.

In view of Lemma~\ref{lem:boundedTimeSW}, if there exists such a play,
there exists a history $h$ such that for all timed configurations
$(\config,\dat)$ along~$h$, $\dat \leq x$ and
$\Players= \Visit(h)$.

During the procedure we use a counter~$S$ that stores the current sum
of the players' costs. We also need to keep in memory the set of
players who have already visited their target vertex; we~write~$I$ for
this set, and $s = (\config,\dat)$ for the current timed
configuration.

The algorithm begins with $S = 0$, $I = \{1\leq i\leq \nbp \mid \src_i=\tgt_i\}$, and $s = (\init,0)$.
Then, step-by-step:
\begin{enumerate}
\item the algorithm guesses the next timed configuration $s' = (\config,\dat)$ with $\dat \leq x$ (if such a sucessor exists);
\item the counter~$S$ is updated for all players $i \not \in I$;
\item $I$ is augmented with $\{1\leq  i\leq\nbp \mid \config(i) = \tgt_i \}$;
\item the algorithm forgets~$s$, and continues from~$s'$,
  unless $I=\Players$.
\end{enumerate}

This procedure may stop for two reasons:
\emph{(i)}~if~at some point $I = \Players$: then we check if $S \leq
x$. If it is the case, then we have found a history~$h$ such that
$\SW(h) \leq x$. This history can be extended to a play~$\rho$ with
$\SW(\rho)\leq x$;
\emph{(ii)}~if~from some timed configuration~$s=(\config,\dat)$, no
successor $(\config',\dat')$ with $\dat' \leq x$ exists. That means
that we have failed to find a history~$h$ such that all players visit
their target vertex within time~$x$, and in particular due to
Lemma~\ref{lem:boundedTimeSW}, there does not exist a play~$\rho$ such
that $\SW(\rho) \leq x$.

The counter~$S$ is incremented by 
at most $\nbp \cdot (x \cdot \max_{\vertex \in \Vertex} \weight(\vertex)(\nbp))$
and there are at most $x$ steps.
Thus the value of~$S$ can always be encoded in polynomial space.
On the other hand, 
the algorithm also has to store two consecutive timed configurations and
the action vector between them; if we assume that the objectives of the players
are given as a function $\Vertex^2\to \bbN$ (with binary encoding),
this takes exponential space. Similarly,
keeping track of which players have visited their objectives
takes exponential space.
\end{proof}

\begin{remark}
Notice that if the objectives are given explicitly as a list of
$(\src,\tgt)$-pairs, one for each player, then the input would be
exponentially larger, and our algorithm would then be in polynomial
space in the size of the input.
\end{remark}

\medskip

\propconstrainedSWsym*

The proof of this result relies on the following notion of
\emph{abstract weighted graph}, which stores the number of active and
winnign players in each vertex:
\begin{definition}[Abstract weighted graph]
The \emph{abstract weighted graph} $\mathcal{W} = (A,B, \tilde{w})$ of a symmetric TNG~$\timedNetGame$ is defined as follows:
\begin{itemize}
\item the set of vertices is
  $A = [0,\nbp]^{\Vertex} \times [0,\nbp]^{\Vertex} \times\bbN$.
  The~first part $[0,\nbp]^{\Vertex}$ gives the number of \emph{active players}
  in each vertex (\IE, those players that have not visited the target vertex~yet); the second component gives the number of \emph{winning players} in each vertex (those who have already visited the target vertex); the last integer is the current time;
\item the set~$B\subseteq A \times A$ is the set of edges of the graph:
  given two vertices
  $a_1=(P^1_A\colon\Vertex\to[0,\nbp], P^1_W\colon \Vertex \to [0,\nbp],\dat_1)$
  and
  $a_2=(P^2_A\colon\Vertex\to[0,\nbp], P^2_W\colon \Vertex \to [0,\nbp],\dat_2)$,
  there is an edge~$(a_1,a_2)$ in~$B$ whenever
  \begin{itemize}
  \item $\dat_1<\dat_2$;
  \item there exist functions $b_A\colon \Edge\to [0,\nbp]$ and $b_W\colon
    \Edge\to [0,\nbp]$, representing the number of active and winning
    players taking each edge of~$\timedNetGame$, such that
    \begin{itemize}
    \item $1\leq \sum_{\edge\in \Edge} (b_A(\edge)+b_W(\edge)) \leq \nbp$
      (only the player(s) proposing the shortest delay will move);
    \item for all~$\edge\in\Edge$, if $b_A(\edge)+b_W(\edge)>0$, then
      $\dat_2\models \guard(\edge)$;
    \item for all~$\vertex\in\Vertex$,
      $\sum_{\edge=(\vertex,\vertex')} b_A(\edge) \leq P^1_A(\vertex)$ and 
      $\sum_{\edge=(\vertex,\vertex')} b_W(\edge) \leq P^1_W(\vertex)$;
    \item $P_A^2(\tgt)=0$ and for all~$\vertex\in\Vertex\setminus\{\tgt\}$,
      $P_A^2(\vertex) = P_A^1(\vertex) - \sum_{\edge=(\vertex,\vertex')} b_A(\edge)
      + \sum_{\edge=(\vertex',\vertex)} b_A(\edge)$;
    \item similarly, for~$P_W^2$:
      $P_W^2(\tgt) = P_W^1(\tgt) - \sum_{\edge=(\tgt,\vertex')} b_W(\edge)
      + \sum_{\edge=(\vertex',\tgt)} (b_W(\edge) + b_A(\edge))$, and for all
      $\vertex\in\Vertex\setminus\{\tgt\}$, 
      $P_W^2(\vertex) = P_W^1(\vertex) - \sum_{\edge=(\vertex,\vertex')} b_W(\edge)
      + \sum_{\edge=(\vertex',\vertex)} b_W(\edge)$.
    \end{itemize}
  \end{itemize}
\item the weight function~$\tilde w$ is defined for each edge~$(a_1,a_2)\in B$
  as
  \[
  \tilde w(a_1,a_2) = \sum_{\vertex\in\Vertex \mid P_A^1(\vertex)\not=0}
    \weight(\vertex)(P_A^1(\vertex)+P_W^1(\vertex))\cdot (\dat_2-\dat_1).
  \]
  \end{itemize}
\end{definition}

\begin{proof}[of Proposition~\ref{prop:constrainedSWsym}]
  The non-deterministic algorithm guesses the successive vertices of a
  path~$p$ in the abstract graph~$\calW$ step-by-step. The constraint
  on the social welfare also gives a bound on the length of the path
  to be guessed, so that this algorithm runs in polynomial space.
\end{proof}

\end{document}